\documentclass[a4paper,12pt]{elsarticle}
\usepackage[export]{adjustbox}
\usepackage[english]{babel}
\usepackage[utf8x]{inputenc}
\usepackage{amsmath}
\usepackage{amssymb}
\usepackage{graphicx}
\usepackage{epstopdf}
\usepackage{float}
\usepackage{hyperref}
\usepackage{amsthm}
\usepackage{subfigure}
\usepackage{bm}
\usepackage{eucal}
\usepackage{mathtools}
\usepackage{caption}
\usepackage{xcolor}
\numberwithin{equation}{section}
\usepackage{geometry}\newtheorem{theorem}{Theorem}

\newtheorem{mydef}{Definition}

\newcommand{\seq}[1]{\left \{ {#1} \right \}}

\numberwithin{equation}{section}
\date{}
\begin{document}
\begin{frontmatter}
\title{Adaptive Estimation for Nonlinear Systems using Reproducing Kernel Hilbert Spaces}
\author[PB]{Parag Bobade\fnref{fn1}\corref{C1}}
\ead{paragb4@vt.edu}

\author{Suprotim Majumdar\fnref{fn2}}
\ead{supro21@vt.edu}

\author{Savio Pereira\fnref{fn3}}
\ead{psavio5@vt.edu}

\author{Andrew J. Kurdila\fnref{fn4}}
\ead{kurdila@vt.edu}

\author{John B. Ferris\fnref{fn5}}
\ead{jbferris@vt.edu}

\cortext[C1]{Corresponding Author}
\fntext[fn1]{Graduate Student, Department of Engineering Science and Mechanics, Virginia Tech}
\fntext[fn2]{Graduate Student, Department of Electrical and Computer Engineering, Virginia Tech} 
\fntext[fn3]{Graduate Student, Department of Mechanical Engineering, Virginia Tech}
\fntext[fn4]{Professor, Department of Mechanical Engineering, Virginia Tech} 
\fntext[fn5]{Associate Professor, Department of Mechanical Engineering, Virginia Tech}

\begin{abstract}
This paper extends a conventional, general framework for online adaptive estimation problems for systems governed by unknown nonlinear ordinary differential equations. The central feature of the theory introduced in this paper represents the unknown function as a member of a reproducing kernel Hilbert space (RKHS) and defines a distributed parameter system (DPS) that governs state estimates and estimates of the unknown function. This paper 1) derives sufficient conditions for the existence and stability of the infinite dimensional online estimation problem, 2) derives existence and stability of finite dimensional approximations of the infinite dimensional approximations, and  3) determines sufficient conditions for the convergence of finite dimensional approximations to the infinite dimensional online estimates. A new condition for  persistency of excitation in a RKHS in terms of its evaluation functionals is introduced in the paper  that enables proof of convergence of the finite dimensional approximations of the unknown function in the RKHS. This paper studies two particular choices of the RKHS, those that are generated by exponential functions and those that are generated by multiscale kernels defined from a multiresolution analysis.

\end{abstract}         
\end{frontmatter}
    %******************************************
	\section{Introduction}
	\label{sec:introduction}
 
 \subsection{Motivation: Road and Terrain Mapping}
 \label{subsec:terrain}
    There  has been a steep rise of interest in the last decade among researchers in academia and the commercial sector in autonomous vehicles and self driving cars. Although adaptive estimation has been studied for some time, applications such as terrain or road mapping continue to challenge researchers to further develop the underlying theory and algorithms in this field. These vehicles are required to sense the environment and navigate surrounding terrain without any human intervention. The environmental sensing capability  of such vehicles  must be  able to navigate off-road conditions  or to respond to other agents in urban settings. As a key ingredient to achieve these goals, it can be critical to have a good {\em a priori} knowledge of the surrounding environment as well as the position and orientation of the vehicle in the environment. 
To collect this data for the construction of terrain maps, mobile vehicles equipped with multiple high bandwidth, high resolution imaging sensors are deployed. The mapping sensors retrieve the terrain data relative to the vehicle  and navigation sensors provide  georeferencing relative to a fixed coordinate system. The geospatial data,  which can include the digital terrain maps acquired from these mobile mapping systems, find applications in emergency response planning and road surface monitoring. Further, to improve the ride and handling characteristic of an autonomous vehicle, it might be necessary that these digital terrain maps have accuracy on a  sub-centimeter scale. 

One of the main areas of improvement in current state of the art terrain modeling technologies is the localization. Since the localization heavily relies on the quality of GPS/GNSS, IMU data, it is important to come up with novel approaches which could fuse the data from multiple sensors to generate the best possible estimate of the environment. Contemporary data acquisition systems used to map the environment generate scattered data sets in time and space. These data sets must be either post-processed or processed online for construction of three dimensional terrain maps. 

Fig.\ref{fig:vehicle1} and Fig.\ref{fig:vehicle2} depict a map building vehicle and trailer developed by some of the authors at Virginia Tech. The system generates experimental observations in the form of data that is scattered in time and space. These data sets have extremely high dimensionality. 
Roughly 180 million scattered data points are collected per minute of data acquisition, which corresponds to a data file of roughly $\mathcal{O}(1GB)$ in size. Current algorithms and software developed in-house post-process the scattered data to generate road and terrain maps. This offline batch computing problem can take many days of computing time to complete. It remains a challenging task to derive a theory and associated algorithms that would enable adaptive or online estimation of terrain maps from such high dimensional, scattered measurements.

This paper introduces a novel theory and associated algorithms that are amenable to observations that take the form of scattered data. The key attribute of the approach is that the unknown function representing the terrain is viewed as an element of a RKHS. The RKHS is constructed in terms of a kernel function $k(\cdot,\cdot): \Omega \times \Omega \rightarrow \mathbb{R}$ where $\Omega \subseteq \mathbb{R}^d$ is the domain over  which scattered measurements are made. 
The kernel $k$ can often be used to define a collection of radial basis functions (RBFs) $k_x(\cdot):=k(x,\cdot)$, each of which is said to be centered at some point $x\in \Omega$. For example, these RBFs might be exponentials, wavelets, or thin plate splines \cite{wendland}. 
By embedding the unknown function that represents the terrain in a RKHS, the new formulation generates a system that constitutes a distributed parameter system. The unknown function, representing map terrain, is the infinite dimensional distributed parameter. 
Although the study of infinite dimensional distributed parameter systems can be substantially more difficult than the study of ODEs, a key result is that stability and convergence of the approach can be established succinctly in many cases.  
Much of the complexity \cite{bsdr1997,bdrr1998} associated with construction of Gelfand triples or the analysis of infinitesimal generators and semigroups that define a DPS can be avoided for many examples of the systems in this paper.
The kernel $k(\cdot,\cdot): \Omega \times \Omega \rightarrow \mathbb{R}$ that defines the RKHS provides a natural collection of bases for approximate estimates of the solution that are based directly on some subset of scattered measurements $\{ x_i \}_{i=1}^n \subset \mathbb{R}^d$. 
% We choose as a basis for the finite dimensional subspace of approximants $H_n \subseteq H$, the $\text{span} \{ k_{x_i}(\cdot) | i=1,\ldots, n \}$ with $k_{x_i}(\cdot):= k(x_i , \cdot)$.
% \space As we discuss in more detail below, each function $k_x(\cdot):=k(x,\cdot)$ is often a radial basis function (RBF) centered at $x\in \Omega$. 
It is typical in applications to select the centers $\{x_i\}_{i=1}^n$ that locate the basis functions from some sub-sample of the locations at which the scattered data is measured. Thus, while we do not study the nuances of such methods, in this paper the formulation provides a natural framework to pose so-called ``basis adaptive methods'' such as in~\cite{dzcf2012} and the references therein. 

While our formulation is motivated by this particular application, it is a  general construction for framing and generalizing some conventional approaches for online adaptive estimation. This framework introduces sufficient conditions that guarantee convergence of estimates in spatial domain $\Omega$ to the unknown function $f$. In contrast, nearly all conventional strategies consider stability and convergence in time alone for some fixed finite dimensional space of $\mathbb{R}^d \times \mathbb{R}^n$, with $n$ the number of parameters used to represent the estimate.  The remainder of this paper studies the existence and uniqueness of solutions, stability, and convergence of approximate solutions for the infinite dimensional adaptive estimation problem defined over an RKHS. The paper concludes with an example of an RKHS adaptive estimation problem for a simple model of map building from vehicles.  The numerical example demonstrates the rate of convergence for finite dimensional models constructed from RBF bases that are centered at a subset of scattered observations. 

\begin{figure}
\centering
\includegraphics[scale=0.75]{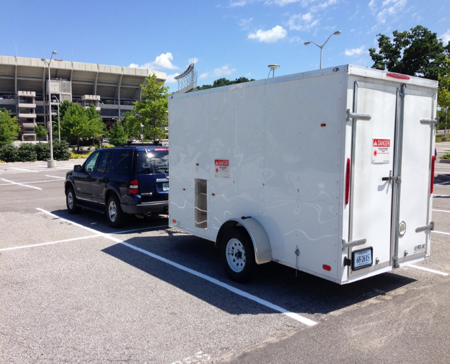}
\caption{Vehicle Terrain Measurement System, Virginia Tech}
\label{fig:vehicle1}
\end{figure}
\begin{figure}
\centering
\includegraphics[scale=0.75]{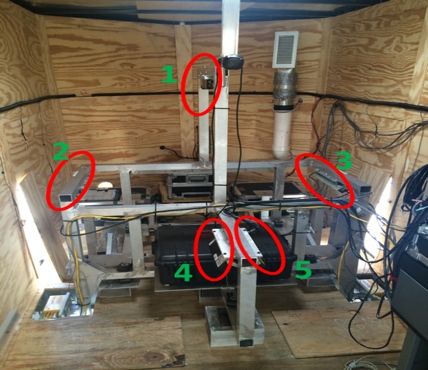}
\caption{Experimental Setup with LMI 3D GO-Locator Lasers}
\label{fig:vehicle2}
\end{figure}  

\subsection{Related Research}
\label{sec:related_research}
The general theory derived in this paper has been motivated in part by the terrain mapping application discussed in Section \ref{sec:introduction}, but also by recent research in a number of fields related to estimation of nonlinear functions.  In this section we briefly review some of the recent research in probabilistic or Bayesian mapping methods, nonlinear approximation and learning theory, statistics, and nonlinear regression.

\subsubsection{Bayesian  and Probabilistic Mapping}
Many popular known techniques adopt a probabilistic approach towards solving the localization and mapping problem in robotics. The algorithms used to solve this problem fundamentally rely on Bayesian estimation techniques like particle filters, Kalman filters and other variants of these methods \cite{Thrun2005Probabilistic, Whyte2006SLAM1, Whyte2006SLAM2}. The computational efforts required to implement these algorithms can be substantial since they involve constructing and updating maps while simultaneously tracking the relative locations of agents with respect to the environment. Over the last three decades significant progress has been made on various frontiers in terms of  high-end sensing capabilities, faster data processing hardwares, robust and efficient computational algorithms \cite{Dissanayake2011Review, Dissanayake2000Computational}. However, the usual Kalman filter based approaches implemented in these applications often are required to address the inconsistency problem in estimation that arise from uncertainties in state estimates \cite{Huang2007Convergence,Julier2001Counter}. Furthermore, it is well acknowledged among the community that these methods suffer from  a major drawback of `{\em closing the loop}'. This refers to the ability to adaptively update the information if it is  revisited. Since such a capability for updating information demands huge memory to store the high resolution and high bandwidth data. Moreover, it is highly nontrivial to guarantee that the  uncertainties in estimates would converge to lower bound at sub optimal rates, since matching these rates and bounds significantly constraint the evolution of states along infeasible trajectories. While probabilistic methods, and in particular Bayesian estimation techniques, for the construction of terrain maps have flourished over the past few decades, relatively few approaches for establishing deterministic theoretical  error bounds in the spatial domain of the unknown function representing the terrain have appeared.

\subsubsection{Approximation and Learning Theory}
Approximation theory has a long history, but the subtopics of most relevance to this paper  include  recent studies in multiresolution analysis (MRA),  radial basis function (RBF) approximation and learning theory. The study of MRA techniques became popular in the late 1980's and early 1990's, and it has flourished since that time. We use only a small part of the general theory of MRAs in this paper,  and we urge the interested reader to consult one of the excellent treatises on this topic for a full account. References \cite{Meyer,mallat,daubechies, dl1993} are good examples of such detailed treatments. We briefly summarize the pertinent aspects of MRA here and in Section \ref{sec:MRA}. A multiresolution analysis defines a family of nested approximation spaces $\seq{H_j}_{j\in \mathbb{N}}\subseteq H$ of an abstract space $H$ in terms of a single function $\phi$, the scaling function.  The approximation space $H_j$  is defined in terms of bases that are constructed from dilates and translates $\seq{\phi_{j,k}}_{k\in \mathbb{Z}^d}$ with $\phi_{j,k}(x):=2^{jd/2}\phi(2^jx-k)$ for $x\in \mathbb{R}^d$ of this single function $\phi$. It is for this reason that  these spaces are sometimes referred to as shift invariant spaces. While the MRA is ordinarily defined only in terms of the scaling functions, the theory provides a rich set of tools to derive bases $\seq{\psi_{j,k}}_{k\in \mathbb{Z}}$, or wavelets, 
for the complement spaces $W_j:=V_{j+1}- V_{j}$.  Our interest in multiresolution analysis arises since these methods can be used to develop multiscale kernels for RKHS, as summarized in  \cite{opfer1,opfer2}. We only consider approximation spaces defined in terms of the scaling functions in this paper. Specifically, with a parameter $s \in \mathbb{R}^+$ measuring smoothness, we use $s-$regular MRAs to define admissible kernels for the reproducing kernels that embody the online and adaptive estimation strategies in this paper.
When the MRA bases are smooth enough, the RKHS kernels derived from a MRA can be shown to be equivalent to a scale of Sobolev spaces having well documented approximation properties. 
The B-spline bases in the numerical examples yield RKHS embeddings with good condition numbers. The details of the RKHS embedding strategy given in terms of wavelet bases associated with an MRA is treated in the forthcoming paper.
\subsubsection{Learning Theory  and Nonlinear Regression}
The  methodology defined in this paper for online adaptive estimation can be viewed as similar in philosophy to the recent efforts that synthesize learning theory and approximation theory. \cite{dkpt2006,kt2007,cdkp2001,t2008} In these references, independent and identically distributed observations of some unknown  function are collected, and they are used to define an estimator of that unknown function.  Sharp estimates of error, guaranteed to hold in probability spaces, are possible using tools familiar from learning theory and thresholding in approximation spaces. The approximation spaces are usually defined   terms of subspaces of an MRA.  However, there are a few key differences between the these efforts in nonlinear regression and learning theory and this paper. The learning theory approaches to estimation of the unknown function depend on observations of the function itself.  In contrast, the adaptive online estimation framework here assumes that observations are made of the estimator states, not directly of the unknown function itself. The learning theory methods also assume a discrete measurement process, instead of the continuous measurement process that characterizes online adaptive estimation. On the other hand, the methods based on learning theory derive sharp function space rates of convergence of the estimates of the unknown function.  Such estimates are not available in conventional online adaptive estimation methods. Typically, convergence in adaptive estimation strategies is guaranteed in time in a fixed finite dimensional space.  One of the significant contributions of this paper is to construct sharp convergence rates in function spaces, similar to approaches in learning theory,  of the unknown function using online adaptive estimation.  

\subsubsection{Online Adaptive Estimation and Control}

Since the approach in this paper generalizes a standard strategy in online adaptive estimation and control theory, we review this class of methods in some detail. This summary will be crucial in understanding the nuances of the proposed technique and in contrasting the sharp estimates of error available in the new strategy to those in the conventional approach. 
    Many popular textbooks study online or adaptive estimation within the context of adaptive control theory for systems governed by ordinary differential equations \cite{sb2012,IaSu,PoFar}. The theory has been extended in several directions, each with its subtle assumptions and associated analyses. 
    Adaptive estimation and control theory has been refined for decades, and significant  progress has been made in deriving convergent estimation and stable control strategies that are robust with respect to some classes of uncertainty.
The efforts in \cite{bsdr1997,bdrr1998} are relevant to this paper, where the authors generalize some of adaptive estimation and model reference adaptive control (MRAC) strategies  for ODEs so that they apply to deterministic infinite dimensional evolution systems. In addition, \cite{dmp1994,dp1988,dpg1991,p1992} also investigate adaptive control and estimation problems under various assumptions for classes of stochastic and infinite dimensional systems. 
Recent developments in $\mathcal{L}^1$ control theory as presented in \cite{HC}, for example, utilize adaptive estimation and control strategies in obtaining stability and convergence for systems generated by collections of nonlinear ODEs. 

To motivate this paper, we consider a model problem in which the plant dynamics are generated by the nonlinear ordinary differential equations
\begin{align}
\dot{x}(t)&= A x(t) + Bf(x(t)), \quad \quad x(0)=x_0 
\label{eq:simple_plant}
\end{align}
with  state $x(t)\in \mathbb{R}^d$, the known  Hurwitz system matrix $ A \in \mathbb{R}^{d\times d}$, the known control influence matrix $B\in \mathbb{R}^d$, and the unknown function $f:\mathbb{R}^d \rightarrow \mathbb{R}$. 
Although this model problem is an exceedingly simple prototypical example studied in adaptive estimation and control of ODEs \cite{sb2012,IaSu,PoFar}, it has proven to be an effective case study in motivating alternative formulations such as in \cite{HC} and will suffice to motivate the current approach. 
Of course, much more general plants are treated in standard methods \cite{sb2012,IaSu,PoFar,naranna} and can be attacked using the strategy that follows.  This structurally simple  problem is chosen so as to clearly illustrate the essential constructions of RKHS embedding method while omitting the nuances associated with general plants. A typical adaptive estimation problem can often be formulated in terms of an estimator equation and a learning law. One of the simplest estimators for this model problem takes the form  
\begin{align}
\dot{\hat{x}}(t)&= A \hat{x}(t) + B\hat{f}(t,x(t)), 
\quad \quad 
\hat{x}(0)=x_0 
\label{eq:sim_estimator}
\end{align}
where $\hat{x}(t)$ is an estimate of the state $x(t)$ and $\hat{f}(t,x(t))$ is time varying estimate of the unknown function $f$ that depends on measurement of the state $x(t)$ of the plant at time $t$. When the state error $\tilde{x}:=x-\hat{x}$ and function estimate error $\tilde{f}:=f-\hat{f}$ are defined, the state error equation is simply 
\begin{align}
\dot{\tilde{x}}(t)&= A \tilde{x}(t) + B\tilde{f}(t,x(t)), \quad \quad 
\tilde{x}(0)=\tilde{x}_0.
\label{eq:sim_error}
\end{align}
The goal of adaptive or online estimation is to determine a learning law that governs the evolution of the function estimate $\hat{f}$ and guarantees that the state estimate $\hat{x}$ converges to the true state $x$,  
$
\tilde{x}(t)= x(t)-\hat{x}(t) \to
0  \text{ as } t\to \infty
$.
Perhaps additionally, it is hoped that the function estimates $\hat{f}$  converge to the unknown function $f$,
$
\tilde{f}(t)= f(t) -\hat{f}(t) \to
0 \text{ as } t \to \infty.
$
The choice of the learning law for the update of the adaptive estimate $\hat{f}$ depends intrinsically on what specific information is available about the unknown function $f$.
It is most often the case for ODEs that the estimate $\hat{f}$ depends on a finite set of unknown parameters $\hat{\alpha}_1,\ldots,\hat{\alpha}_n$.  The learning law is then expressed as an evolution law for the parameters $\hat{\alpha}_i$, $i=1,\ldots,n$. The discussion that follows  emphasizes that this is a very specific underlying assumption regarding the information available about unknown function $f$. Much more general prior assumptions are possible.  
%In the next section, we discuss various  types of knowledge regarding the unknown function $f$.

\subsubsection{Classes of Uncertainty in Adaptive Estimation}
% \label{subsec:uncertainity}
The adaptive estimation task seeks to construct a learning law based on the knowledge that is available regarding the function $f$. 
Different methods for solving this problem have been developed depending on the type of information available about the unknown function $f$. 
The uncertainty about $f$ is often described as forming  a continuum between structured and unstructured uncertainty. 
In the most general case, we might know that $f$ lies in some compact set $\mathcal{C}$ of a  particular Hilbert space of functions $H$ over a subset $\Omega \subseteq \mathbb{R}^d$.
This case, that reflects in some sense the least information regarding the unknown function, can be expressed as the condition that 
$
f \in \{g \in \mathcal{C} | \mathcal{C}\subset {H} \},
$
 for some compact set of functions $\mathcal{C}$ in a Hilbert space of functions $H$. 
In approximation theory, learning theory, or non-parametric estimation problems this information is sometimes referred to as the {\em prior}, and choices of $H$ commonly known as the hypothesis space. The selection of the hypothesis space $H$ and set $\mathcal{C}$ often reflect the approximation,  smoothness, or compactness properties of the unknown function \cite{dkpt2006}.
This example may in some sense utilize only limited or minimal information regarding the unknown function $f$, and we may refer to the uncertainty as unstructured. Numerous variants of conventional adaptive estimation admit additional knowledge about the unknown function. 
In most conventional cases the unknown function $f$ is assumed to be given in terms of some fixed set of parameters. 
This situation is similar in philosophy to problems of parametric estimation which restrict approximants to classes of functions that admit representation in terms of a specific set of parameters. 
Suppose the finite dimensional basis $\left \{ \phi_k\right \}_{k=1,\ldots, n}$ is known  for a particular finite dimensional subspace $H_n \subseteq H$ in which the function lies, and further that the uncertainty is expressed as the condition that there is a unique set of unknown coefficients $\left \{\alpha_i^*\right\}_{i=1,\ldots,n} $ such that $f:=f^*=\sum_{i=1,\ldots,n} \alpha_i^* \phi_i \in H_n$.  Consequently, conventional approaches may restrict the adaptive estimation technique to construct an estimate with knowledge that $f$ lies in the set
\begin{align}
\label{eq:e2}
f \in \biggl \{ g \in H_n \subseteq H \biggl | 
&g = \sum_{i=1,\ldots,n} \alpha_i \phi_i
\text{ with } \\
\notag &\alpha_i \in [a_i,b_i] 
\subset \mathbb{R} \text{ for } i=1,\ldots,n 
\biggr \}
\end{align}
\noindent This is an example where the uncertainty in the estimation problem may be said to be structured.  The unknown function is parameterized by the collection of coefficients $\{\alpha_i^*\}_{i=1,\ldots,n}$.
In this case the compact set the $\mathcal{C}$ is  a subset of $H_n$. As we discuss in sections ~\ref{subsec:Lit},~\ref{sec:RKHS},and ~\ref{sec:existence}, the RKHS embedding approach can be characterised by the fact that the uncertainty is more general and even unstructured, in contrast to conventional methods.
% Clearly in this case the compact set $\mathcal{C}\subset H_n$.
%
%%

\subsubsection{Adaptive Estimation in $\mathbb{R}^d \times \mathbb{R}^n$}
\label{subsec:adapt1}
The development of adaptive estimation strategies when the uncertainty takes the form in \ref{eq:e2} represents, in some sense, an iconic approach in the adaptive estimation and control community.  
Entire volumes \cite{sb2012,IaSu,PoFar,NarPar199D} contain numerous variants of strategies that can be applied to solve adaptive estimation problems in which the uncertainty takes the form in \ref{eq:e2}. 
One canonical approach to such an  adaptive estimation problem is  governed by three coupled equations:  the plant dynamics ~\ref{eq:f}, estimator equation \ref{eq:a2}, and the learning rule.
We organize the basis functions as $\phi:=[\phi_1,\dots,\phi_n]^T$ and the parameters as $\alpha^{*^T}=[\alpha^*_1,\ldots,\alpha^*_n]$, 
$\hat{\alpha}^T=[\hat{\alpha}_1,\ldots,\hat{\alpha}_n]$. A common gradient based learning law yields the governing equations that incorporate the plant dynamics,   estimator equation, and the learning rule.
\begin{align}
\label{eq:f}
\dot{x}(t) &= Ax(t) + B \alpha^{*^T} \phi(x(t)),\\
\label{eq:a2}
\dot{\hat{x}}(t) &
=A \hat{x}(t) + B  \hat{\alpha}^T(t) \phi(x(t)), \\
\label{eq:a3}
\dot{\hat{\alpha}}(t) &= \Gamma^{-1}\phi B^T P(x-\hat{x}),
\end{align}
where   $\Gamma\in \mathbb{R}^{n\times n}$ is  symmetric and positive definite. The symmetric positive definite matrix $P\in\mathbb{R}^{d\times d}$ is the unique solution of Lyapunov's equation $A^T P + PA = -Q$, for some selected symmetric positive definite $Q \in \mathbb{R}^{d\times d}$.
\noindent Usually the above equations are summarized in terms the two error equations 
\begin{align}
\label{eq:a4}
\dot{\tilde{x}}(t) &= A \tilde{x} + B \phi^{T}(x(t))\tilde{\alpha}(t)\\
\label{eq:a5}
\dot{\tilde{\alpha}}(t) &= -\Gamma^{-1} \phi(x(t)) B^T P\tilde{x}.
\end{align}
with $\tilde{\alpha}:=\alpha^*-\hat{\alpha}$ and $\tilde{x}:=x-\hat{x}$. 
 Equations ~\ref{eq:a4},~\ref{eq:a5} can also be written as 
\begin{align}
\begin{Bmatrix}
\dot{\tilde{x}}(t) \\
\dot{\tilde{\alpha}}(t)
\end{Bmatrix}
=
\begin{bmatrix}
A & B \phi^T (x(t))\\
-\Gamma^{-1} \phi(x(t)) B ^T P & 0
\end{bmatrix}
\begin{Bmatrix}
\tilde{x}(t)\\
\tilde{\alpha}(t)
\end{Bmatrix}.
\label{eq:error_conv}
\end{align}
This equation defines an evolution on $\mathbb{R}^d \times \mathbb{R}^n$
and has been studied in great detail in ~\cite{naranna,narkud,mornar}. 
Standard texts such as ~\cite{sb2012,IaSu,PoFar,NarPar199D} outline numerous other variants for the online adaptive estimation problem using projection, least squares methods and other popular approaches.
% See ~\cite{bdrr1998,bsdr1997} for infinite dimensional systems.
%********************************************************************

%********************************************************************

      \subsection{Overview of Our Results}
      \label{subsec:Lit}
 \subsubsection{Adaptive Estimation in $\mathbb{R}^d \times H$}
 \label{subsec:adapt2}
 In this paper, we study the method of RKHS embedding that interprets the unknown function $f$ as an element of the RKHS $H$, without any {\em a priori} selection of the particular finite dimensional subspace used for estimation of the unknown function.  The counterparts to Equations ~\ref{eq:f},~\ref{eq:a2},~\ref{eq:a3} are  the plant, estimator, and learning laws 
\begin{align}
\dot{x}(t) &= Ax(t) + BE_{x(t)}f,\\
\dot{\hat{x}}(t) &= A\hat{x}(t) + BE_{x(t)}\hat{f}(t), \label{eq:rkhs_plant}\\
\dot{\hat{f}}(t) &= \Gamma^{-1}(BE_{x(t)})^*P(x(t) - \hat{x}(t)),
\end{align}
where as before $x,\hat{x}\in \mathbb{R}^d$, but $f$ and $\hat{f}(t)\in H$,  $E_{\xi}: H \to \mathbb{R}^d $ is the evaluation functional  that is given by  $E_{\xi}: f \mapsto f(\xi)$ for  all $\xi\in \mathbb{R}^d$ and $f \in H$, and $\Gamma\in \mathcal{L}(H,H)$ is a self adjoint, positive definite linear  operator.a  The error equation  analogous to Equation~\ref{eq:error_conv} system is then given by
\begin{align}
\begin{Bmatrix}
\dot{\tilde{x}}(t) \\
\dot{\tilde{f}}(t)
\end{Bmatrix}
=
\begin{bmatrix}
A & B E_{x(t)}\\
-\Gamma^{-1}(B E_{x(t)})^*P & 0
\end{bmatrix}
\begin{Bmatrix}
\tilde{x}(t)\\
\tilde{f}(t)
\end{Bmatrix},
\label{eq:eom_rkhs}
\end{align}
which defines an evolution on $\mathbb{R}^d \times H$, instead of on $\mathbb{R}^d \times \mathbb{R}^n$. 
%**********

\subsubsection{Existence, Stability, and Convergence Rates}
We briefly summarize and compare the conlusions that can be reached for the conventional and RKHS embedding approaches. Let $(\hat{x}, \hat{f})$ be estimates of $(x,f)$ that evolve according to  the state, estimator, and learning law of RKHS embedding. Define the state and distributed parameter error as  $\tilde{x}:=x-\hat{x}$ and $\tilde{f}:=f-\hat{f}$, respectively.  Under the assumptions outlined in Theorems \ref{th:unique}, \ref{th:stability}, and \ref{th:PE} for each $T>0$ there is a unique mild solution for the error $(\tilde{x},\tilde{f})\in C([0,T];\mathbb{R}^d\times H)$  to the DPS described by Equations \ref{eq:eom_rkhs}. Moreover, the error in state estimates $\tilde{x}(t)$ converges to zero,
$\lim_{t \rightarrow \infty} \| \tilde{x}(t)\|=0$. If all the evolutions with initial conditions in an open ball containing the origin  exist in $C([0,\infty);\mathbb{R}\times H)$,  the equilibrium at the origin $(\tilde{x},\tilde{f})=(0,0)$ is stable. The results so far are therefore entirely analogous to conventional estimation method, but are cast in the infinite dimensional RKHS $H$. See the standard texts ~\cite{sb2012,IaSu,PoFar,NarPar199D} for proofs of existence and convergence of the conventional methods. It must be emphasized again that the conventional results are stated for evolutions in $\mathbb{R}^d\times\mathbb{R}^n$, and the RKHS results hold for evolutions in $\mathbb{R}^d\times H$. Considerably more can be said about the convergence of finite dimensional approximations. For the RKHS embedding approach state  and finite dimensional approximations $(\hat{x}_j,\hat{f}_j)$ of the infinite dimensional estimates $(\hat{x},\hat{f})$ on a grid  that has  resolution level $j$ are governed by Equations \ref{eq:approx_on_est1} and \ref{eq:approx_on_est2}. The finite dimensional estimates $(\hat{x}_j,\hat{f}_j)$ converge to the infinite dimensional estimates $(\hat{x},\hat{f})$ at a rate that depends on $\|I-\Gamma\Pi_j^*\Gamma_j^{-1} \Pi_j\|$ and $\|I - \Pi_j\|$ where $\Pi_j : H \to H_j$ is the $H$-orthogonal projection.

The remainder of this paper studies the existence and uniqueness of solutions, stability, and convergence of approximate solutions for infinite dimensional, online or adaptive estimation problems. The analysis is based on a study of  distributed parameter systems (DPS) that contains the RKHS $H$.  The paper concludes with an example of an RKHS adaptive estimation problem for a simple model of map building from vehicles.  The numerical example demonstrates the rate of convergence for finite dimensional models constructed from radial basis function (RBF) bases that are centered at a subset of scattered observations. 
The discussion focuses on a comparison and contrast of the analysis for the ODE system and the distributed parameter system.
Prior to these discussions, however, we present a brief review fundamental properties of RKHS spaces in the next section.
 \section{Reproducing Kernel Hilbert Space}
 \label{sec:RKHS}
Estimation techniques for distributed parameter systems have been previously studied in \cite{bk1989}, and further developed to incorporate adaptive estimation of parameters in certain infinite dimensional systems by \cite{bsdr1997} and the references therein. These works also presented the necessary conditions required to achieve parameter convergence during online estimation. But both approaches rely on delicate semigroup analysis and evolution, or Gelfand triples.The approach herein is much simpler and amenable to a wide class of applications. It appears to be simpler, practical approach to generalise conventional methods. This paper considers estimation problems that are cast in terms of the unknown function $f:\Omega \subseteq \mathbb{R}^d \to \mathbb{R}$, and our approximations will assume that this function is an element of a reproducing kernel Hilbert space. One way to define a reproducing kernel Hilbert space relies on demonstrating the boundedness of evaluation functionals, but we briefly summarize a constructive approach that is helpful in applications and understanding computations such as in our numerical examples. 

In this paper $\mathbb{R}$ denotes the real numbers, $\mathbb{N}$ the positive integers, $\mathbb{N}_0$ the non-negative integers, and $\mathbb{Z}$ the integers. We follow the convention that $a \gtrsim b$ means that there is a constant $c$, independent of $a$ or $b$, such that $b \leq ca$. When $a\gtrsim b $ and $b\gtrsim a$, we write $a \approx b $. Several function spaces are used in this paper. The $p$-integrable Lebesgue spaces are denoted $L^p(\Omega)$ for $1\leq p \leq \infty$, and $C^s (\Omega)$ is the space of continuous functions on $\Omega$ all of whose derivatives less than or equal to $s$ are continuous. The space $C_b^s (\Omega)$ is the normed vector subspace of $C^s (\Omega)$ and consists of all $f\in C^s (\Omega)$ whose derivatives of order less than or equal to $s$ are bounded. The space $C^{s,\lambda} (\Omega)\subseteq C_b^s (\Omega) \subseteq C^s (\Omega)$ is the collection of functions  with derivatives $\frac{\partial^{|\alpha|}f}{\partial x^{|\alpha|}}$ that are $\lambda$-Holder continuous, 
\begin{align*}
\|f(x)-f(y)\| \leq C\|x - y\|^{\lambda}
\end{align*}
The Sobolev space of functions that have weak derivatives of the order less than equal to $r$  that lie in $L^p(\Omega)$ is denoted $H^r_p(\Omega)$.

A reproducing kernel Hilbert space is constructed in terms of a symmetric,  continuous, and positive definite function $k:\Omega \times \Omega \to \mathbb{R}$, where positive definiteness requires that for any finite collection of points 
$\{x_i\}_{i=1}^n \subseteq \Omega $ 
$$\sum_{i,j=1}^{n}k(x_i , x_j ) \alpha_i \alpha_j \gtrsim \|\alpha\|^{2}_{\mathbb{R}^n}
$$
for all $\alpha = \{\alpha_1,\hdots, \alpha_n \}^T$.. For each $x\in \Omega$, we denote the function $k_x := k_x (\cdot) = k(x,\cdot)$ and refer to $k_x$ as the kernel function centered at $x$. In many typical examples ~\cite{wendland}, $k_x$ can be interpreted literally as a radial basis function centered at $x\in \Omega$. For any kernel functions $k_x$ and $k_y$ centered at $x,y \in \Omega$,  we define the  inner product $(k_x,k_y):= k(x,y)$. 
The RKHS $H$ is then defined as the completion of all finite sums extracted from the set $\{k_x|x \in \Omega\}$.
It is well known that this construction guarantees the boundedness of the evaluation functionals $E_x : H \to \mathbb{R}$. In other words for each $x\in \Omega$ we have a constant $c_x$ such that 
$$ |E_x f | = |f(x)| \leq c_x \|f\|_H$$
for all $f\in H$. The reproducing property of the RKHS $H$ plays a crucial role in the analysis here, and it states that,
$$E_xf = f(x) = (k_x , f)_H$$
for $x \in \Omega$ and $f\in H$. We will also require the adjoint $E_x^* :\mathbb{R}\to H $ in this paper, which can be calculated directly by noting that 
$$ (E_x f,\alpha )_\mathbb{R} = (f,\alpha k_x)_H = (f,E_x^* \alpha)_H $$
for $\alpha \in \mathbb{R}$ , $x\in \Omega$ and $f\in H$. Hence, $E_x^* : \alpha \mapsto \alpha k_x \in H$.

Finally, we will be interested in the specific case in which it is possible to show that the RKHS $H$ is a subset of $C(\Omega)$, and furthermore, that the associated injection$i:H \rightarrow C(\Omega)$ is uniformly bounded.
This uniform embedding is possible, for example, provided that the kernel is bounded by a constant $\tilde{C}^2$, 
$
\sup_{x\in \Omega} k(x,x) \leq  \tilde{C}^2.
$
This fact follows by first noting that by the reproducing kernel property of the RKHS, 
we  can write 
\begin{equation}
|f(x)|=|E_x f |= |(k_x, f)_H | \leq \|k_x \|_H \|f\|_H.
\end{equation}
From the definition of the inner product on $H$, we have 
$
\|k_x \|^2=|(k_x, k_x)_H |=|(k(x,x)| \leq \tilde{C}^2.
$
It follows that $\|if\|_{C(\Omega)}:= \|f\|_{C(\Omega)} \leq {\tilde{C}} \|f\|_H$ and thereby that $\|i\|\leq {\tilde{C}}$. We next give two examples that will be studied in this paper.

\subsection*{Example: The Exponential Kernel}
A popular example of an RKHS, one that will be used in the numerical examples, is constructed from the family of exponentials $\kappa(x,y):=e^{-\| x-y\|^2/\sigma^2}$ where $\sigma>0$. 
Suppose that  $\tilde{C} = \sqrt{\sup_{x\in\Omega}\kappa(x,x)}<\infty$. Smale and Zhou in \cite{sz2007} argue that 
$$
|f(x)|=|E_x(f)|=|(\kappa_x,f)_H|\leq 
\|\kappa_x\|_H \|f\|_H
$$
for all $x\in \Omega$ and $f\in H$, and since 
$\|\kappa_x\|^2=|\kappa(x,x)|\leq \tilde{C}^2$,  it follows  that the embedding $i:H \rightarrow L^\infty(\Omega)$ is bounded,
$$
\|f\|_{L^\infty(\Omega)}:=\|i(f)\|_{L^\infty(\Omega)}\leq \tilde{C} \|f\|_H.
$$
For the exponential kernel above, $\tilde{C}=1$. 
Let $C^s(\Omega)$ denote the space of functions on $\Omega$ all of whose partial derivatives of order less than or equal to $s$ are continuous.  The space $C^s_b(\Omega)$ is endowed with the norm
$$
\|f\|_{C^s_b(\Omega)}:= \max_{|\alpha|\leq s}
\left \|  
\frac{\partial^{|\alpha|}f}{\partial x^\alpha}
\right \|_{L^\infty(\Omega)},
$$
with the summation  taken over  multi-indices $\alpha:=\left \{ \alpha_1, \ldots,\alpha_d \right \}\in \mathbb{N}^d$, $\partial x^{\alpha}:=\partial x_1^{\alpha_1} \cdots \partial x_d^{\alpha_d}$, and  $|\alpha|=\sum_{i=1,\ldots,d} \alpha_i$. 
Observe that the continuous functions in $C^s(\Omega)$ need not be bounded even if $\Omega$ is a bounded open domain. The space $C^s_b(\Omega)$ is the subspace consisting of functions  $f\in C^s_b(\Omega)$ for which all derivatives of order less than or equal to $s$ are bounded.  
The space $C^{s,\lambda}(\Omega)$ is the subspace of functions $f$ in  $C^{s}(\Omega)$ 
for which all of the partial derivatives $\frac{\partial f^{|\alpha|}}{\partial x^\alpha}$ with $|\alpha|\le s$ are
$\lambda$-Holder continuous. The norm of $C^{s,\lambda}(\Omega)$ for $0 < \lambda \leq 1$  is given by
$$
\|f\|_{C^{s,\lambda}(\Omega)} = \|f\|_{C^s(\Omega)}+ \max_{0 \leq \alpha \leq s} \sup_{\substack{x,y\in \Omega \\x\ne y}}\frac{\left| \frac{\partial^{|\alpha|} f}{\partial x^{|\alpha|}}(x) -\frac{\partial^{|\alpha|}f}{\partial x^{|\alpha|}}(y) \right|}{|x-y|^\lambda}
$$
Also, reference \cite{sz2007} notes that if $\kappa(\cdot,\cdot)\in C^{2s,\lambda}_b(\Omega \times \Omega)$ with $0<\lambda<2$ and $\Omega$ is a closed domain, then the inclusion $H\rightarrow C^{s,\lambda/2}_b(\Omega)$ is well defined and continuous.  That is the mapping $i:H \rightarrow C^{s,\lambda/2}_b$ defined via $f\mapsto i(f):=f$ satisfies
$$
\| f\|_{C^{s,\lambda/2}_b(\Omega)}\lesssim \|f\|_H.
$$
In fact reference \cite{sz2007} shows that 
$$
\|f \|_{C^s_b(\Omega)} \leq 4^s \|\kappa\|_{{C^{2s}_b}(\Omega\times \Omega)}^{1/2} \|f\|_H.
$$
The overall important conclusion to draw from the summary above is that there are many conditions that guarantee that the imbedding $H\hookrightarrow C_b(\Omega)$ is continuous. This condition will play a central role in devising simple conditions for existence of solutions of the RKHS embedding technique.

%****************%*****************%*******

\subsection{Multiscale Kernels Induced by $s$-Regular Scaling Functions}
\label{sec:MRA}
The characterization of the norm of the Sobolev space $H^{r}_2:=H^{r}_2(\mathbb{R}^d)$ has appeared in many monographs that discuss multiresolution analysis \cite{Meyer,mallat,devore1998}. It is also possible to define the Sobolev space $H^{r}_2(\mathbb{R}^d)$ as the Hilbert space constructed from a reproducing kernel $\kappa(\cdot,\cdot):\mathbb{R}^d \times \mathbb{R}^d \rightarrow \mathbb{R}$ that is defined in terms of an $s$-regular scaling function $\phi$ of an multi-resolution analysis (MRA) \cite{Meyer,devore1998}. The scaling function $\phi$ is $s$-regular provided that, for $\frac{d}{2}<r<s$, we define the kernel
\begin{align*}
\kappa(u,v):&=\sum_{j=0}^\infty 2^{j(d-2r)}\sum_{k\in \mathbb{Z}^d}\phi(2^ju-k)\phi(2^jv-k)\\
&=\sum_{j=0}^\infty 2^{-2rj}\sum_{k\in \mathbb{Z}^d}\phi_{j,k}(u)\phi_{j,k}(v)
.\end{align*}
It should be  noted that  the requirement  $d/2<r$ implies  the coefficient $2^{j(d-2r)}$ above is  decreasing as $j\rightarrow \infty$, and ensures  the summation converges. As discussed in Section \ref{sec:RKHS}  and in reference \cite{opfer1,opfer2}, the RKHS is constructed as the closure of the finite linear span of the set of function $\left\{\kappa_u\right\}_{u\in \Omega}$ with $\kappa_u(\cdot):=\kappa(u,\cdot)$. Under the assumption that $\frac{d}{2}<r<s$, the Sobolev space $H^r_2(\mathbb{R}^d)$  
can also be related to the Hilbert space $H_\kappa^r(\mathbb{R}^d)$ 
defined as 
\begin{align*}
H_{\kappa}^r(\mathbb{R}^d):=\left\{ f:\mathbb{R}^d\rightarrow\mathbb{R} \mid (f,f)_{\kappa,r}^\frac{1}{2}=\|f\|_{\kappa,r}<\infty\right\}
\end{align*}
with the inner product $(\cdot,\cdot)_{\kappa,r}$ on $H_{\kappa}^r(\mathbb{R}^d)$ defined as  
\begin{align*}
(f,f)_{\kappa,r}&:=\|f\|_{\kappa,r}^2:= 
\inf \biggl\{ \sum_{j=0}^\infty  2^{j(2r-d)}\|f_j\|_{V_j}^2\biggl|
 f_j\in V_j, f=\sum_{j=0}^\infty f_j\biggr\}
\end{align*}
with $\|f\|^2_{V_j}=\sum_{k \in \mathbb{Z}^d} c_{j,k}^2 $ for  $f_j(u)=\sum_{k \in \mathbb{Z}^d}c_{j,k}\phi(2^ju-k)$ and $j\in \mathbb{N}_0$. Note that the characterization above of $H_{\kappa}^r(\mathbb{R}^d)$ is expressed only in terms of the scaling functions $\phi_{j,k}$ for $j\in \mathbb{N}_0$ and $k\in \mathbb{Z}^d$. The functions $\phi$ and $\psi$ need not define an orthonormal multiresolution in this characterization, and the bases $\psi_{j,k}$ for the complement spaces $W_j$ are not used. We discuss the use of wavelet bases $\psi_{j,k}$ for the definition of the kernel in forthcoming paper. References \cite{opfer1,opfer2} show that when $d/2< r<s$, we have the norm equivalence  
\begin{align}
H_{\kappa}^r(\mathbb{R}^d)\approx H^{r}_2(\mathbb{R}^d).
\label{eq:norm_equiv}
\end{align}
Finally, from Sobolev's Embedding Theorem \cite{af2003}, whenever $r>d/2$ we have the embedding
$$
H^r_2 \hookrightarrow C_b^{r-d/2} \subset C^{r-d/2}
$$
where $C_b^r$ is the subspace of functions $f$ in  $C^r$ all of whose derivatives up through order $r$ are bounded. In fact, by choosing the $s$-regular MRA with $s$ and $r$ large enough, we have the imbedding 
$H^r_2(\Omega) \hookrightarrow C(\Omega)$ when $\Omega \subseteq \mathbb{R}^d$ \cite{af2003}.

One of the simplest examples that meet the conditions of this section includes the normalized B-splines of order $r>0$. We denote by $N^r$ the normalized B-spline of order $r$ with integer knots and define its translated dilates by  $N^r_{j,k}:=2^{jd/2}N^r(2^{jd} x - k)$ for $k\in \mathbb{Z}^d$ and $j\in \mathbb{N}_0$. In this case the kernel is written in the form
$$
\kappa(u,v):=\sum_{j=0}^\infty 2^{-2rj}\sum_{k\in \mathbb{Z}^d}N^r_{j,k}(u)N^r_{j,k}(v).
$$
Figure \ref{fig:nbsplines} depicts the translated dilates of the normalized B-splines of order $1$ and $2$ respectively.
\begin{center}
\begin{figure}[h!]
\centering
\begin{tabular}{cc}
\includegraphics[width=.4\textwidth]{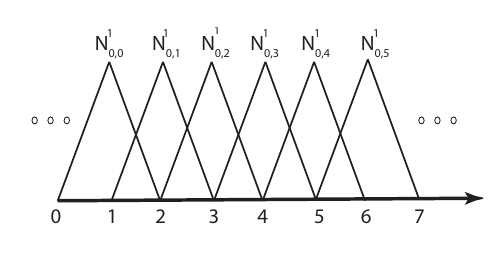}
&
\includegraphics[width=.4\textwidth]{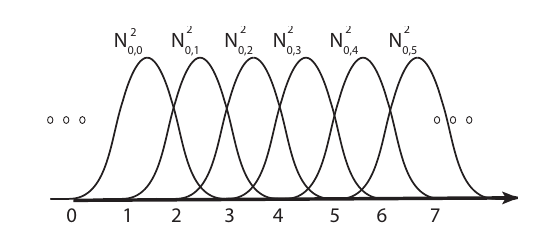}\\
{ B-splines $N^1$}
&
{ B-splines $N^2$}
\end{tabular}
\caption{Translated Dilates of Normalized B-Splines}
\label{fig:nbsplines}
\end{figure}
\end{center}
%************************
\section{Existence,Uniqueness and Stability}
\label{sec:existence}
In the adaptive estimation problem  that is cast in terms of a RKHS $H$, we seek a solution $X = (\tilde{x},\tilde{f}) \in \mathbb{R}^d \times H \equiv \mathbb{X}$ that satisfies Equation \ref{eq:eom_rkhs}. 
In general $\mathbb{X}$ is an infinite dimensional state space for this estimation problem, which can in principle substantially complicate the analysis in comparison to conventional ODE methods. 
We first establish that the adaptive estimation problem in Equation \ref{eq:eom_rkhs} is well-posed. 
The result that is derived below is not the most general possible, but rather has been emphasised because its conditions are  simple and easily verifiable in many applications.
\begin{theorem}
\label{th:unique}
Suppose that $x \in C([0,T];\mathbb{R}^d)$ and that the embedding $i:H \hookrightarrow C(\Omega)$ is uniform in the sense that there is a constant $C>0$ such that for any $f \in H$,
\begin{equation}
\label{6}
\|f\|_{C(\Omega)}\equiv \|if\|_{C(\Omega)} \leq C\|f\|_H.
\end{equation}
For any $T>0$ there is a unique mild solution $(\tilde{X},\tilde{f}) \in C([0,T],\mathbb{X})$ to Equation \ref{eq:eom_rkhs} and the map $X_0 \equiv (\tilde{x}_0,\tilde{f}_0) \mapsto (\tilde{x},\tilde{f}) $ is Lipschitz continuous from $\mathbb{X}$ to $C([0,T],\mathbb{X})$.
\end{theorem}
\begin{proof}
We can split the governing Equation \ref{eq:eom_rkhs} into the form
\begin{align}
\begin{split}
\begin{Bmatrix}
\dot{\tilde{x}}(t)\\
\dot{{\tilde{f}}}(t)
\end{Bmatrix}
= 
&\begin{bmatrix}
A & 0\\
0 & A_0
\end{bmatrix}
\begin{Bmatrix}
\tilde{x}(t)\\
\tilde{f}(t)
\end{Bmatrix}+
\begin{bmatrix}
0 & B E_{(x(t))}\\
-\Gamma^1 (B E_{(x(t)})^* P & -A_0
\end{bmatrix}
\begin{Bmatrix}
\tilde{x}(t)\\
\tilde{f}(t)
\end{Bmatrix},
\end{split}
\end{align}
and write it more concisely as 
\begin{equation}
\dot{\tilde{X}} = \mathbb{A}\tilde{X}(t) + \mathbb{F}(t,\tilde{X}(t))
\end{equation}
where the operator $A_0 \in \mathcal{L}(H,H)$ is arbitrary. It is immediately clear that $\mathbb{A}$ is the infinitesimal generator of $C_0$ semigroup on $\mathbb{X}\equiv \mathbb{R}^d\times H$ since $\mathbb{A}$ is bounded on $\mathbb{X}$. In addition, we see the following:
\begin{enumerate} 
\item  The function $\mathbb{F}: \mathbb{R}^+ \times \mathbb{X} \to \mathbb{X}$ is uniformly globally Lipschitz continuous: there is a constant $L>0$ such that  
$$
\|\mathbb{F}(t,X)-\mathbb{F}(t,Y)\| \leq L\|X-Y\|
$$ 
for all $ X,Y \in \mathbb{X}$ and $t\in [0,T]$. 
\item The map $t \mapsto \mathbb{F}(t,X)$ is continuous on $[0,T]$ for each fixed $X\in \mathbb{X}$.
\end{enumerate}
By Theorem 1.2, p.184, in reference \cite{pazy},  there is a unique mild solution 
$$\tilde{X} = \{\tilde{x},\tilde{f}\}^T \in C([0,T];\mathbb{X})\equiv C([0,T];\mathbb{R}^d\times H). $$
In fact the map $\tilde{X}_0 \mapsto X$ is Lipschitz continuous from $\mathbb{X}\to C([0,T];\mathbb{X})$.
\end{proof}
The proof of stability of the equilibrium at the origin of the RKHS 
Equation \ref{eq:eom_rkhs} closely resembles the Lyapunov analysis of Equation \ref{eq:error_conv}; the extension to consideration of the infinite dimensional state space $\mathbb{X}$ is required.
It is useful to carry out this analysis in some detail to see how the adjoint $E_x^* :\mathbb{R}\to H $ of the evaluation functional $E_x : H \to \mathbb{R}$ plays a central and indispensable role in the study of the stability of evolution equations on the RKHS.
\begin{theorem}
\label{th:stability}
Suppose that the RKHS Equations \ref{eq:eom_rkhs} have a unique solution in $C([0,\infty);H)$ for every initial condition $X_0$ in some open ball $B_r (0) \subseteq \mathbb{X}$. Then the equilibrium at the origin  is Lyapunov stable. Moreover, the state error $\tilde{x}(t) \rightarrow 0$ as $t \rightarrow \infty$.  
\end{theorem}
\begin{proof}
Define the Lyapunov function $V:\mathbb{X} \to \mathbb{R}$ as 
$$ V \begin{Bmatrix}
\tilde{x}\\
\tilde{f}
\end{Bmatrix}
= \frac{1}{2}\tilde{x}^T P\tilde{x} + \frac{1}{2}(\Gamma \tilde{f},\tilde{f})_H.
$$
This function is norm continuous and positive definite on any neighborhood of the origin since $ V(X) \geq \|X\|^2_{\mathbb{X}}$ for all $X \in \mathbb{X}$. For any $X$, and in particular over the open set $B_r(0)$, the derivative of the Lyapunov function $V$ along trajectories of the system is given as 
\begin{align*}
\dot{V} &= \frac{1}{2}(\dot{\tilde{x}}^T P\tilde{x}+\tilde{x}^TP\dot{\tilde{x}})+(\Gamma \tilde{f},\dot{\tilde{f}})_H\\
&= -\frac{1}{2}\tilde{x}^T Q\tilde{x}+(\tilde{f},E_x^*B^*P\tilde{x}+\Gamma\dot{\tilde{f}})_{H}= -\frac{1}{2}\tilde{x}^T Q\tilde{x},
\end{align*}
since $(\tilde{f},E_x^*B^*P\tilde{x}+\Gamma\dot{\tilde{f}})_{H}=0$.
Let $\epsilon$ be some constant such that $0 < \epsilon < r$. Define $\gamma (\epsilon)$ and $\Omega_\gamma$ according to 
$$\gamma(\epsilon) = \inf_{\|X\|_\mathbb{X}=\epsilon} V(X),$$
$$\Omega_\gamma = \{X \in \mathbb{X}|V(X)<\gamma \}.$$
We can picture these quantities as shown in Fig. \ref{fig:lyapfun} and Fig. \ref{fig:kernels}.
\begin{figure}
\centering
\includegraphics[scale=0.35]{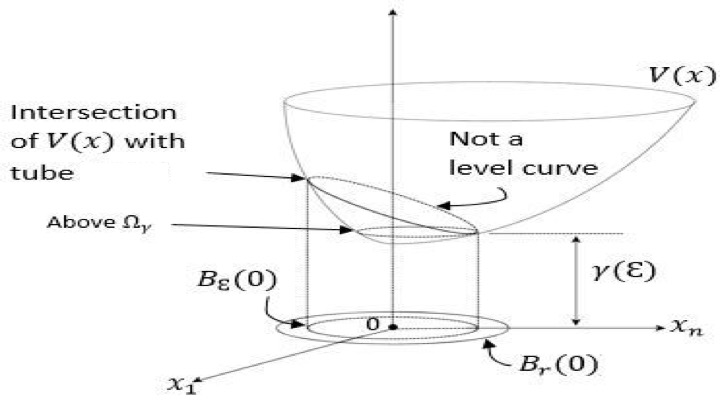}
\caption{Lyapunov function, $V(x)$}
\label{fig:lyapfun}
\end{figure}
\begin{figure}
\centering
\includegraphics[scale=0.55]{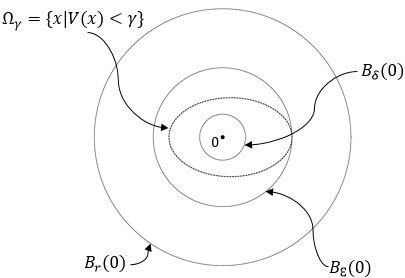}
\caption{Stability of the equilibrium}
\label{fig:kernels}
\end{figure}
But $\Omega_\gamma=\{X\in \mathbb{X}|V(X)<\gamma\}$ is an open set since it  is the inverse image of the open set $(-\infty,\gamma) \subset \mathbb{R}$ under the continuous mapping $V:\mathbb{X} \to \mathbb{R}$. The set $\Omega_\gamma$ therefore contains an open neighborhood of each of its elements.  Let $\delta>0$ be the radius of such an open ball containing the origin with  $B_\delta(0) \subset \Omega_\gamma$. 
Since $\overline{\Omega}_\gamma:=\{X\in \mathbb{X}|V(X)\leq \gamma\}$ is a  level set of $V$ and $V$ is non-increasing, it is a positive invariant set. Given any initial condition $x_0 \in B_\delta(0) \subseteq \Omega_\gamma$, we know that the trajectory $x(t)$ starting at $x_0$ satisfies
$x(t) \in \overline{\Omega}_\gamma \subseteq \overline{B_\epsilon(0)}  \subseteq B_r(0)$  for all $t\in [0,\infty)$. 
The equilibrium at the origin is stable.

The convergence of the state estimation error $\tilde{x}(t) \rightarrow 0$ as $t\rightarrow \infty$  can be based on Barbalat's lemma by modifying the conventional arguments for ODE systems.  Since $\frac{d}{dt}(V(X(t))) = - \frac{1}{2} \tilde{x}^T(t) Q \tilde{x}\leq 0$,  $V(X(t))$ is non-increasing and bounded below by zero. There is a constant $V_\infty:=\lim_{t \rightarrow \infty}V(X(t))$, and we have
$$
V(X_0)-V_\infty = \int_0^\infty \tilde{x}^T(\tau)Q\tilde{x} d\tau \gtrsim \|\tilde{x}\|^2_{L^2((0,\infty);\mathbb{R}^d)}.
$$
Since $V(X(t)) \leq V(X_0)$, we likewise have $\|\tilde{x}\|_{L^\infty(0,\infty)}\lesssim V(X_0)$ and $\|\tilde{f}\|_{L^\infty((0,\infty);H)}\lesssim V(X_0)$. The equation of motion enables a uniform bound on $\dot{\tilde{x}}$ since 
\begin{align}
&\|\dot{\tilde{x}}(t)\|_{\mathbb{R}^d}
\leq \|A\| \| \tilde{x}(t)\|_{\mathbb{R}^d}
+ \|B\| \|E_{x(t)} \tilde{f}(t)\|_{\mathbb{R}^d}, \notag \\
&\leq \|A\| \| \tilde{x}(t)\|_{\mathbb{R}^d}
+ \tilde{C} \|B\|  \| \tilde{f}(t) \|_{H},\\
& \leq \|A\| \|\tilde{x}\|_{L^\infty((0,\infty);\mathbb{R}^d)} 
+ \tilde{C} \|B\|  \| \tilde{f} \|_{L^\infty((0,\infty),H)}. \notag
\end{align}
Since $\tilde{x}\in L^\infty((0,\infty);\mathbb{R}^d)) \cap L^2((0,\infty);\mathbb{R}^d)$ and $\dot{\tilde{x}} \in L^\infty((0,\infty);\mathbb{R}^d)$, we conclude by generalizations of  Barbalat's lemma \cite{Farkas2016Variations} that $\tilde{x}(t) \rightarrow 0$ as $t \to \infty$.
\end{proof}

It is evident that Theorem \ref{th:stability} yields results about stability and convergence over the RKHS of the state estimate error to zero that are analogous to typical results for conventional ODE systems. As expected, conclusions for the convergence of the function estimates $\hat{f}$ to $f$ are more difficult to generate, and they rely on {\em persistency of excitation } conditions that are suitably extended to the RKHS framework.
\begin{mydef}
We say that the plant in the RKHS Equation ~\ref{eq:rkhs_plant} is {\em strongly persistently exciting} if there exist constants $\Delta,\gamma>0,\text{ and }T$ such that for $f\in H$ with $\|f\|_H=1$ and $t>T$ sufficiently large, 
$$
\int_{t}^{t+\Delta}
\left(E^*_{x(\tau)}E_{x(\tau)}f,f\right)_H d\tau \gtrsim \gamma.
$$
\end{mydef}
As in the consideration of ODE systems, persistency of excitation is sufficient to guarantee convergence of the function parameter estimates to the true function.

\begin{theorem}
\label{th:PE}
Suppose that the plant in Equation \ref{eq:rkhs_plant} is strongly persistently exciting and that either (i) the function $k(x(.),x(.)) \in L^1((0,\infty);\mathbb{R})$, or (ii) the matrix $-A$ is coercive in the sense that $(-Av,v)\geq c\|v\|^2$ $\forall$ $v\in\mathbb{R}^d$ and $\Gamma =P=I_d$. Then the parameter function error $\tilde{f}$ converges strongly to zero,
$$
\lim_{t\rightarrow \infty} \| f-\hat{f}(t) \|_H = 0.
$$
\end{theorem}
\begin{proof}
We begin by assuming $(i)$ holds, 
In the proof of Theorem \ref{th:stability} it is shown that $V$ is bounded below and non-increasing, and therefore approaches a limit
$$
\lim_{t\rightarrow \infty} V(t)=V_\infty< \infty.
$$
Since $\tilde{x}(t) \rightarrow 0$ as $t\rightarrow \infty$, we can conclude that the limit
$$
\lim_{t\rightarrow \infty} \| \tilde{f}(t) \|_H \lesssim V_\infty.
$$
Suppose that $V_\infty \not = 0.$ Then there exists a positive, increasing sequence of  times $\left\{ t_k\right \}_{k\in \mathbb{N}}$  with $\lim_{k\rightarrow \infty} t_k = \infty$ and  some constant $\delta>0$ 
such that 
$$
\| \tilde{f}(t_k)\|^2_H \ge \delta
$$
for all $k\in\mathbb{N}$. 
Since the RKHS is persistently exciting, we can write
\begin{align*}
\int^{t_k+\Delta}_{t_k} \left(E^{*}_{x(\tau)}E_{x(\tau)}\tilde{f}(t_k),\tilde{f}(t_k)\right)_Hd\tau \gtrsim \gamma \| \tilde{f}{(t_k)}\|_{H}^{2} \geq \gamma \delta
\end{align*}

for each $k\in \mathbb{N}$. By the reproducing property of the RKHS, we can then see that 
\begin{align*}
\gamma \delta \leq \gamma \| \tilde{f}(t_k) \|_H^2 &\lesssim \int_{t_k}^{t_k + \Delta} \left ( \kappa_{x(\tau)}, \tilde{f}(t_k) \right )_H^2 d\tau\\
&\leq \|\tilde{f}(t_k)\|_H^2 \int_{t_k}^{t_k + \Delta} \|\kappa_{x(\tau)} \|_H^2 d\tau \\
&= \| \tilde{f}(t_k) \|_H^2
\int_{t_k}^{t_k+\Delta} \left (\kappa_{x(\tau)},\kappa_{x(\tau)}\right )_H d\tau \\
& = \| \tilde{f}(t_k) \|_H^2 
\int_{t_k}^{t_k+\Delta} \kappa(x(\tau),x(\tau)) d\tau.
\end{align*}
Since $\kappa_r(x(.),x(.)) \in L^1((0,\infty);\mathbb{R})$ by assumption, when we take the limit as $k\rightarrow \infty$, we obtain the contradiction $0<\gamma \leq 0$.  We conclude therefore that $V_\infty=0$ and $\lim_{t\rightarrow \infty} \|\tilde{f}(t)\|_H = 0$. 

We outline the proof when (ii) holds, which is based on slight modifications of arguments that appear in \cite{d1993,bsdr1997,dr1994,dr1994pe,bdrr1998,kr1994} that treat a different class of infinite dimensional nonlinear systems whose state space is cast in terms of a Gelfand triple. 
Perhaps the simplest analysis follows from \cite{bsdr1997} for this case. Our hypothesis that $\Gamma=P=I_d$ reduces Equations \ref{eq:eom_rkhs} to the form of Equations 2.20 in \cite{bsdr1997}. The assumption that $-A$ is coercive in our theorem implies the coercivity assumption (A4) in \cite{bsdr1997} holds. If we define $\mathbb{X}=\mathbb{Y}:=\mathbb{R}^n \times H$, then it is clear that the imbeddings $\mathbb{Y} \rightarrow \mathbb{X} \rightarrow \mathbb{Y}$ are continuous and dense, so that they define a  Gelfand triple.  Because of the trivial form of the Gelfand triple in this case, it is immediate that the Garding inequality holds in Equation 2.17 in \cite{bsdr1997}.  
 We identify $BE_{x(t)}$ as the control influence operator $\mathcal{B}^*(\overline{u}(t))$ in \cite{bsdr1997}.
Under these conditions, Theorem ~\ref{th:PE} follows from Theorem 3.4 in \cite{bsdr1997} as a special case.
\end{proof}
% Perhaps the simplest analysis for $(ii)$ follows when:
% \begin{enumerate}
% \item We consider the degenerate case  with $H\equiv V$ in the Gelfand triple as in \cite{bsdr1997}.
% \item We identify $E_{x(t)}$ as the control influence operator $\mathcal{B}^*(\overline{u}(t))$ in \cite{bsdr1997}.
% \end{enumerate}
% Under these conditions, Theorem ~\ref{th:PE} follows from Theorem 3.4 in \cite{bsdr1997} as a special case.
% \end{proof}
   \section{Finite Dimensional Approximations}
   \label{sec:finite}
\subsection{Convergence of Finite Dimensional Approximations}
The governing system in Equations \ref{eq:eom_rkhs} constitute a distributed parameter system since the functions $\tilde{f}(t)$ evolve in the infinite dimensional space $H$.  In practice these equations must be approximated by some finite dimensional system.  Let $\{H_n\}_{n\in\mathbb{N}_0} \subseteq H$ be a nested sequence of subspaces. Let $\Pi_j$ be a collection of approximation operators $
\Pi_j:{H}\rightarrow {H}_n$ such that  $\lim_{j\to \infty}\Pi_j f = f$ for all $f\in H$ and $\sup_{j\in \mathbb{N}_0} \|\Pi_j\| \leq C $ for a constant $C > 0$. Perhaps the most evident example of such collection might choose $\Pi_j$ as the $H$-orthogonal projection for a dense collection of subspaces $H_n$. It is also common to choose $\Pi_j$ as a uniformly bounded family of quasi-interpolants \cite{devore1998}. We next construct a finite dimensional approximations $\hat{x}_j$ and $\hat{f}_j$ of the online estimation equations in
\begin{align}
\dot{\hat{x}}_j(t)  & = A\hat{x}_j(t) + 
B E_{x(t)} \Pi^*_j  \hat{f}_j(t), \label{eq:approx_on_est1} \\
\dot{\hat{f}}_j(t) & =  \Gamma_j^{-1}\left ( B E_{x(t)} \Pi^*_j  \right)^* P\tilde{x}_j(t)
\label{eq:approx_on_est2}
\end{align}
with $\tilde{x}_j:=x-\hat{x}_j$. 
It is important to note that in the above equation $
\Pi_j:{H}\rightarrow {H}_n$, and $\Pi_j^*:{H}_n\rightarrow {H}$.
\begin{theorem}
Suppose that $x \in C([0,T],\mathbb{R}^d)$ and that the embedding $i:H \to C(\Omega)$ is uniform in the sense that 
\begin{equation}
\label{6}
\|f\|_{C(\Omega)}\equiv \|if\|_{C(\Omega)} \leq C\|f\|_H.
\end{equation}
Then for any $T>0$, 
\begin{align*}
\| \hat{x} - \hat{x}_j\|_{C([0,T];\mathbb{R}^d)} &\rightarrow 0,\\
\|\hat{f} - \hat{f}_j\|_{C([0,T];H)} &\rightarrow 0,
\end{align*}
as $j\rightarrow \infty$.
\end{theorem}
\begin{proof}
Define the operators $\Lambda(t):= B E_{x(t)}:H\rightarrow \mathbb{R}^d$ and for each $t\geq 0$, introduce the measures of state estimation error $\overline{x}_j:=\hat{x}-\hat{x}_j$, and define the function estimation error $\overline{f}_j
=\hat{f}-\hat{f}_j$. 
Note that $\tilde{x}_j:=x-\hat{x}_j=x-\hat{x} + \hat{x}-\hat{x}_j=\tilde{x}+ \overline{x}_j$.
The time derivative of the error induced by  approximation of the estimates can be expanded as follows:
\begin{align*}
&\frac{1}{2} \frac{d}{dt}\left (
( {\overline{x}}_j, {\overline{x}}_j )_{\mathbb{R}^d} + ({\overline{f}}_j,{\overline{f}}_j   )_H 
\right )  = 
( \dot{\overline{x}}_j, {\overline{x}}_j )_{\mathbb{R}^d} + (\dot{\overline{f}}_j,{\overline{f}}_j   )_H 
   \\
&= (A\overline{x}_j + \Lambda \overline{f}_j , \overline{x}_j)_{\mathbb{R}^d} + 
\left ( 
\left (\Gamma^{-1}-\Pi_j^*\Gamma_j^{-1}\Pi_j \right )
\Lambda^*P \tilde{x}, \overline{f}_j
\right )_H 
-\left (\Pi_j^* \Gamma_j^{-1} \Pi_j \Lambda^* P \overline{x}_j,\overline{f}_j \right)_H
\\
&\leq C_A \| \overline{x}_j \|^2_{\mathbb{R}^d} + \|\Lambda\| \| \overline{f}_j \|_{H} \| \overline{x}_j \|_{\mathbb{R}^d} \\
&\quad \quad 
+ \| \Gamma^{-1} 
(I-\Gamma \Pi_j^*\Gamma_j^{-1}\Pi_j) \Lambda^* P \tilde{x}\|_{H} \|\overline{f}_j \|_H 
+\left \| 
\Pi_j^* \Gamma_j^{-1} \Pi_j \Lambda^* P  
\right \| \|\overline{x}_j\| 
\|\overline{f}_j \|
\\
& \leq 
C_A \| \overline{x}_j \|_{\mathbb{R}^d}^2 + \frac{1}{2}
\|\Lambda\| \left ( 
\| \overline{f}_j \|_{H}^2
+ \| \overline{x}_j \|_{\mathbb{R}^d}^2
\right ) 
+ \frac{1}{2}\|\Pi^*_j \Gamma_j^{-1} \Pi_j\|
\| \Lambda^*\| \|P\| \left (  \|\overline{x}_j\|^2_{\mathbb{R}^d} + \| \overline{f}_j\|_H \right ) 
\\
&\quad \quad
+ \frac{1}{2} \left (  
\Gamma^{-1} 
(I-\Gamma \Pi_j^*\Gamma_j^{-1}\Pi_j) \Lambda^* P \tilde{x}\|_{H}
+ 
\|\overline{f}_j \|^2_H 
\right) \\
& \leq
\frac{1}{2} \|\Gamma^{-1} \| \| \Lambda^*\| \|P\|
\| I-\Gamma \Pi_j^*\Gamma_j^{-1}\Pi_j \|^2\|\tilde{x}\|^2_{\mathbb{R}^d}
+\\
&\quad \quad
+\left (C_A + \frac{1}{2} \|\Lambda\| 
+ \frac{1}{2} C_B \|\Lambda^*\| \|P\|
\right ) \|\overline{x}_j\|^{2}_{\mathbb{R}^d}
+
\frac{1}{2} \left ( \|\Lambda\| + 1
+ \frac{1}{2} C_B \|\Lambda^*\| \|P\|\right) \|\overline{f}_j\|^{2}_H 
\end{align*}
We know that $\|\Lambda(t)\|=\|\Lambda^*(t)\|$ is bounded uniformly in time from the assumption that  $H$  is uniformly embedded in $C(\Omega)$.
We next consider the operator error that manifests in the term $(\Gamma^{-1} - \Pi^*_j \Gamma_j^{-1} \Pi_j)$. For any $g\in H$ we have
\begin{align*}
\| (\Gamma^{-1} - \Pi^*_j \Gamma_j^{-1} \Pi_j)g \|_H & =
\| \Gamma^{-1}( I - \Gamma \Pi^*_j \Gamma_j^{-1} \Pi_j)g \|_H \\
&\leq 
\| \Gamma^{-1} \|  
\|\left (\Pi_j + (I-\Pi_j)\right )( I - \Gamma \Pi^*_j \Gamma_j^{-1} \Pi_j)g \|_H \\
&\lesssim \| I-\Pi_j \| \|g\|_H.
\end{align*}
This final inequality follows since $\Pi_j(I - \Gamma \Pi^*_j \Gamma_j^{-1} \Pi_j)=0$ and 
$\Gamma \Pi^*_j \Gamma_j^{-1} \Pi_j\equiv\Gamma \Pi^*_j \left (\Pi_j \Gamma \Pi_j^* \right)^{-1} \Pi_j $ is uniformly bounded.
We then  can write
\begin{align*}
 \frac{d}{dt}\left (
 \|\overline{x}_j\|^2_{\mathbb{R}^d} + \|\overline{f}_j\|^2_H  
\right )
&\leq C_1 \| I-\Gamma \Pi_j^*\Gamma_j^{-1}\Pi_j \|^2 \\
&\quad \quad+ C_2 \left (\|\overline{x}_j\|^2_{\mathbb{R}^d} + \|\overline{f}_j\|^2_H \right )
\end{align*}
where $C_1,C_2>0$. We integrate this inequality over the interval $[0,T]$ and obtain
\begin{align*}
\|\overline{x}_j(t)\|^2_{\mathbb{R}^d}
+ \|\overline{f}_j(t)\|^2_H 
&\leq 
\|\overline{x}_j(0)\|^2_{\mathbb{R}^d}
+ \|\overline{f}_j(0)\|^2_H \\
&
+ C_1T  \| I-\Gamma\Pi_j^*\Gamma_j^{-1}\Pi_j \|^2 \\
&+ C_2\int_0^T \left ( 
\|\overline{x}_j(\tau)\|^2_{\mathbb{R}^d}
+ \|\overline{f}_j(\tau)\|^2_H
\right ) d\tau
\end{align*}
We can always choose $\hat{x}(0) = \hat{x}_j(0)$, so that $\overline{x}_j(0) = 0$. If we choose $\hat{f}_j(0):=\Pi_j\hat{f}(0)$ then,
\begin{align*}
\|\overline{f}_j(0)\| &= \|\hat{f}(0)-\Pi_j\hat{f}(0)\|_H\\
&\leq \|I-\Pi_j\|_H \|\hat{f}(0)\|_H.
\end{align*}
The non-decreasing term can be rewritten as $C_1T  \| I-\Gamma\Pi_j^* \Gamma_j^{-1} \Pi_j \|^2 \leq C_3 \|I-\Pi_j\|^2_H$. 
\begin{align}
\|\overline{x}_j(t)\|^2_{\mathbb{R}^d}
+ \|\overline{f}_j(t)\|^2_H 
&\leq  C_4\|I-\Pi_j\|^2_H+ C_2\int_0^T \left ( 
\|\overline{x}_j(\tau)\|^2_{\mathbb{R}^d}
+ \|\overline{f}_j(\tau)\|^2_H
\right ) d\tau
\label{eq:gron_last}
\end{align}
Let  $\alpha(t):=C_4\|I-\Pi_j\|^2_H$ and applying Gronwall's inequality to equation \ref{eq:gron_last}, we get
\begin{align}
\|\overline{x}_j(t)\|^2_{\mathbb{R}^d}
+ \|\overline{f}_j(t)\|^2_H 
&\leq \alpha(t) e^{C_2 T}
\end{align}
As $j\to \infty$ we get $\alpha(t) \to 0$, this implies $\overline{x}_j(t)\to 0$ and $\overline{f}_j(t)\to 0$.
Therefore the finite dimensional approximation converges to the infinite dimensional states in $\mathbb{R}^d \times H$.
\end{proof} 
% Lastly, the finite dimensional matrix form of the learning law is given by
% \begin{align*}
% \dot{\alpha}_i=G_{ij}^{-1}\phi_j(x(t)) B^TP\tilde{x}.
% \end{align*}
% where $G_{ij}= (\phi_i,\phi_j)_H$ is the Grammian matrix. 
	\section{Numerical Simulations}
    \label{sec:numerical}
\begin{figure}
\centering
\includegraphics[scale=0.3]{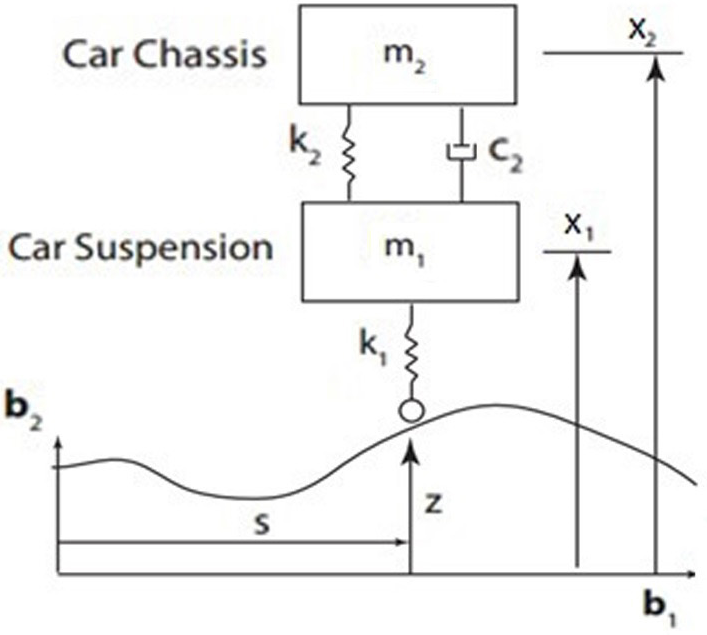}
\hspace{1cm}
\includegraphics[scale=0.3]{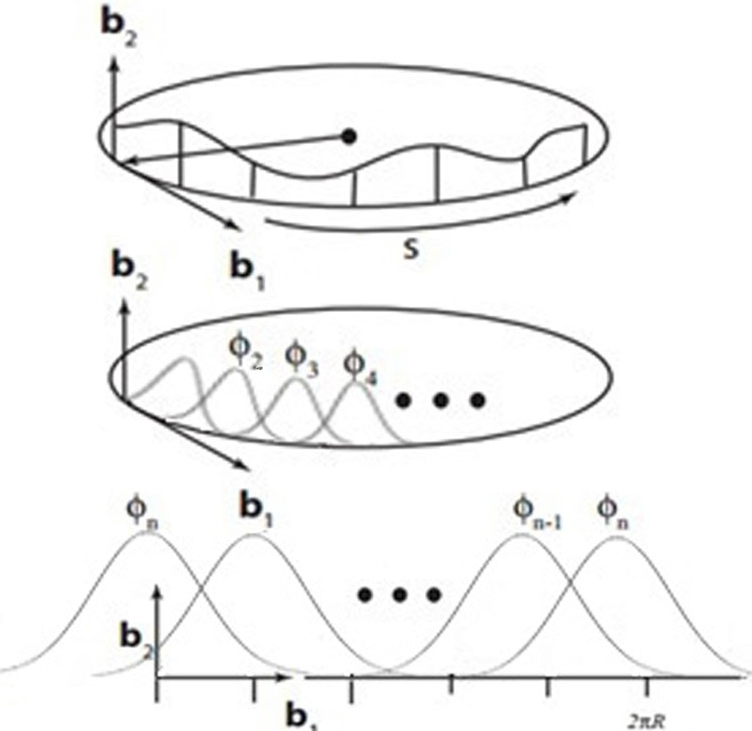}
\captionsetup{justification=justified,margin=1cm}
\caption{Experimental setup and definition of basis functions}
\label{fig:Model}
\end{figure}    
A schematic representation of a quarter car model consisting of a chassis, suspension and road measuring device is shown in Fig ~\ref{fig:Model}. In this simple model the displacement of car suspension and chassis are $x_1$ and $x_2$ respectively. The arc length $s$ measures the distance along the track that vehicle follows. The equation of motion for the two DOF model has the form,
\begin{equation}
M\ddot{x}(t)+C\dot{x}(t)+Kx(t)=Bf(s(t))
\end{equation}
with the mass matrix $M \in \mathbb{R}^{2\times2}$, the stiffness matrix $K \in \mathbb{R}^{2\times2}$, the damping matrix $C \in \mathbb{R}^{2\times2}$, the control influence vector $b \in \mathbb{R}^{2\times 1}$ in this example. The road profile is denoted by the unknown function $f:\mathbb{R} \to \mathbb{R}$. For simulation purposes, the car is assumed to traverse a circular path of radius $R$, so that we restrict attention to periodic round profiles $f : [0,R]\to \mathbb{R}$. To illustrate the methodology, we first assume that the unknown function, $f$ is  restricted to the class of uncertainty mentioned in Equation~\ref{eq:e2} and therefore can be approximated as
\begin{equation}
f(\cdot)=\sum_{i=1}^n{\alpha_i^*k_{x_i}(\cdot)}
\end{equation}
with $n$ as the number of basis functions, $\alpha_i^*$ are the  true unknown coefficients to be estimated, and $k_{x_i}(\cdot)$ are basis functions over the circular domain. 
Hence the state space equation can be written in the form
\begin{equation}
\dot{x}(t)=Ax(t)+B\sum_{i=1}^n{\alpha_i^*k_{x_i}(s(t))}.
\label{eq:num_sim}
\end{equation}
where the state vector $x = [\dot{x}_1,x_1,\dot{x}_2,x_2]$, the system matrix $A\in \mathbb{R}^{4 \times 4}$, and control influence matrix $B \in \mathbb{R}^{4 \times 1}$.
For the quarter car model shown in Fig. \ref{fig:Model} we derive the matrices, 
$$
A=\begin{bmatrix}
\frac{-c_2}{m_1} &\frac{-(k_1+k_2)}{m_1} &\frac{c_2}{m_1} &\frac{k_2}{m_1}\\
1 &0 &0 &0\\
\frac{-c_2}{m_2} &\frac{(k_2)}{m_2} &\frac{-c_2}{m_2} &\frac{-k_2}{m_2}\\
0 &0 &1 &0
\end{bmatrix}
\quad \text{and} \quad 
B=\begin{bmatrix}
\frac{k_1}{m_1}\\
0\\
0\\
0
\end{bmatrix}.
$$
Note that if we augment the state to be $\{x_1,x_2,x_3,x_4,s\}$ and append an ODE that specifies $\dot{s}(t)$ for $t\in \mathbb{R}^+$ the equations ~\ref{eq:num_sim}  can be written in the form of equations ~\ref{eq:simple_plant}.Then the finite dimensional set of coupled ODE's for the adaptive estimation problem can be written in terms  of the plant dynamics, estimator equation, and the learning law which are of the form shown in Equations \ref{eq:f}, \ref{eq:a2}, and \ref{eq:a3} respectively.

%*****************************************************************%
% 	\section{Results}
% 	\label{sec:results}
    \subsection{Synthetic Road Profile}
    The constants in the equation are initialized as follows:  $m_1=0.5$ kg, $m_2=0.5$ kg, $k_1=50000$ N/m, $k_2=30000$ N/m and $c_2=200$ Ns/m, $\Gamma=0.001$. 
The radius of the path traversed $R=4$ m, the road profile to be estimated is assumed to have the shape $f(\cdot)= \kappa\sin(2\pi \nu (\cdot))$ where $\nu =0.04$ Hz and $\kappa=2$.  
Thus our adaptive estimation problem is formulated  for a synthetic road profile in the RKHS $H = \overline{\{k_x(\cdot)|x\in \Omega\}}$ with  $k_x(\cdot)=e^\frac{-\|x-{\cdot} \|^2}{2\sigma^2 }$.
The radial basis functions, each with standard deviation of $\sigma=50$, span over the range of $25^o$ with their centers $s_i$ evenly separated along the arc length. It is important to note that we have chosen a scattered basis that can be located at any collection of centers $\{s_i\}_{i=1}^{n}\subseteq \Omega$ but the uniformly spaced centers are selected to illustrate the convergence rates. 
\begin{figure}[h!]
\centering
\includegraphics[scale=0.45]{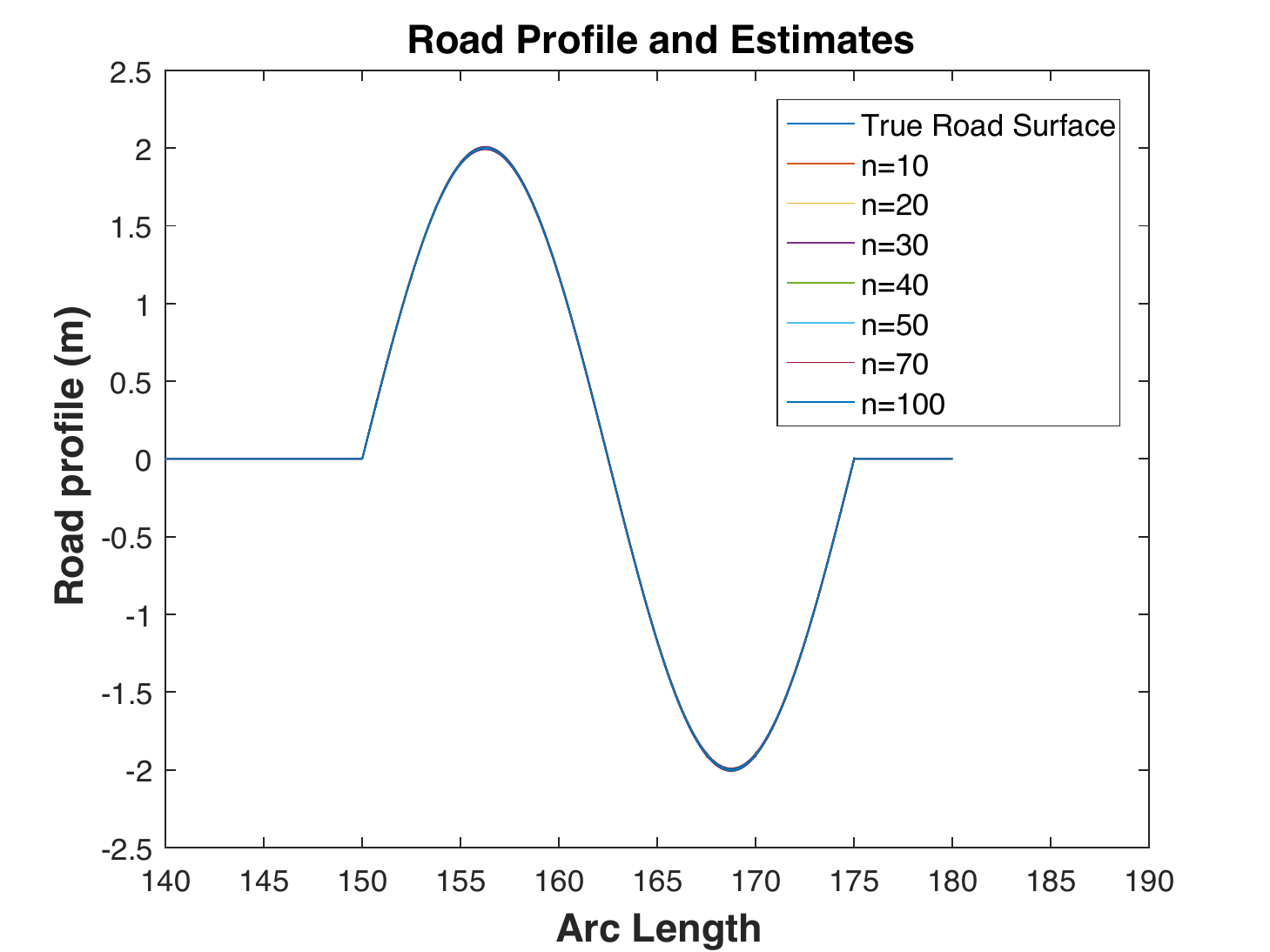}
\caption{Road surface estimates for $n=\{10,20,\cdots,100\}$}
\label{fig:Sine Road}
\end{figure}
Fig.\ref{fig:Sine Road} shows the finite dimensional estimates $\hat{f}$ of the road and the true road surface $f$ for different number of basis kernels ranging from $n=\{10,20,\cdots,100\}$. 
\begin{figure}[h!]
\centering
\begin{tabular}{cc}
\includegraphics[width=.5\textwidth]{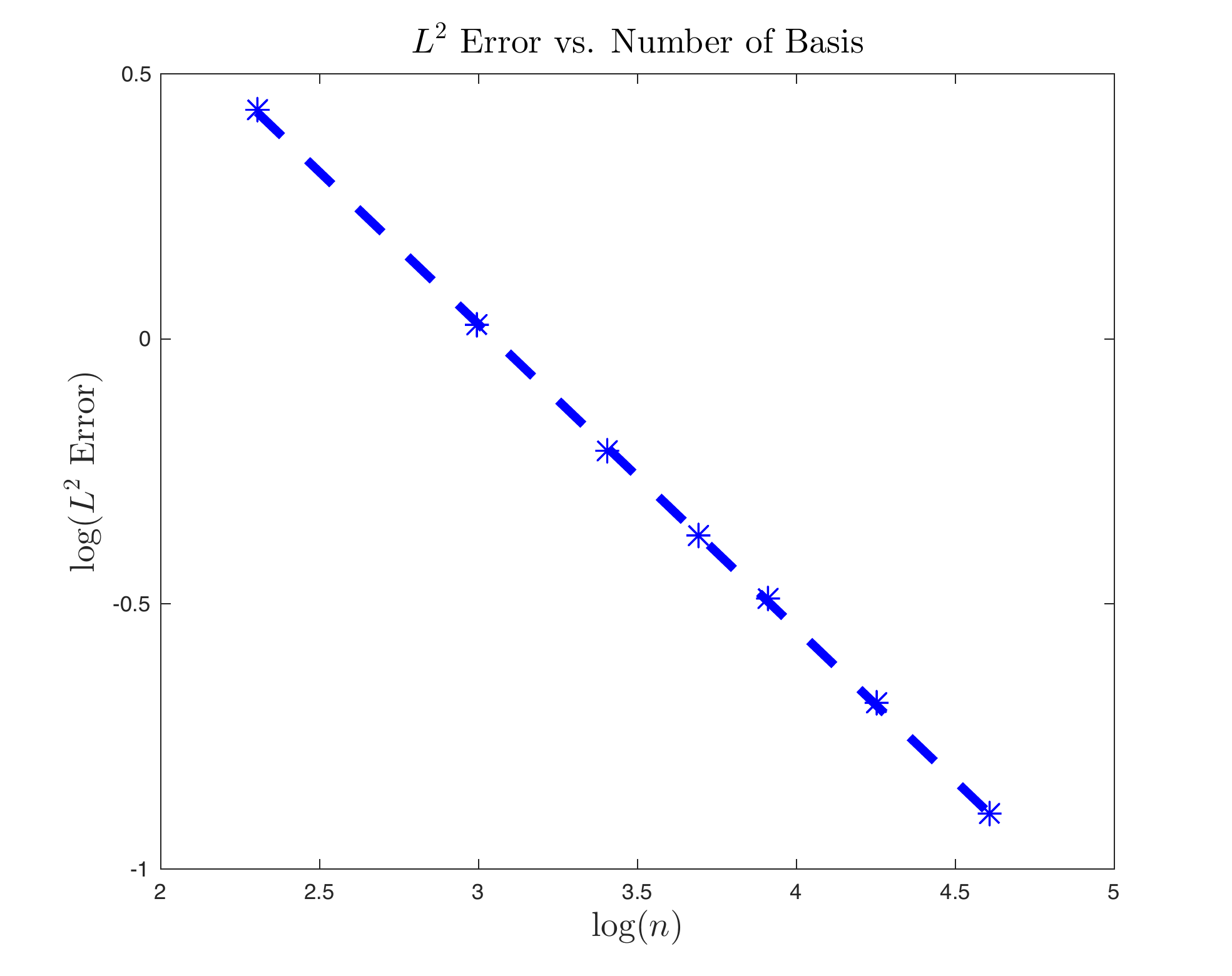}
&
\includegraphics[width=.5\textwidth]{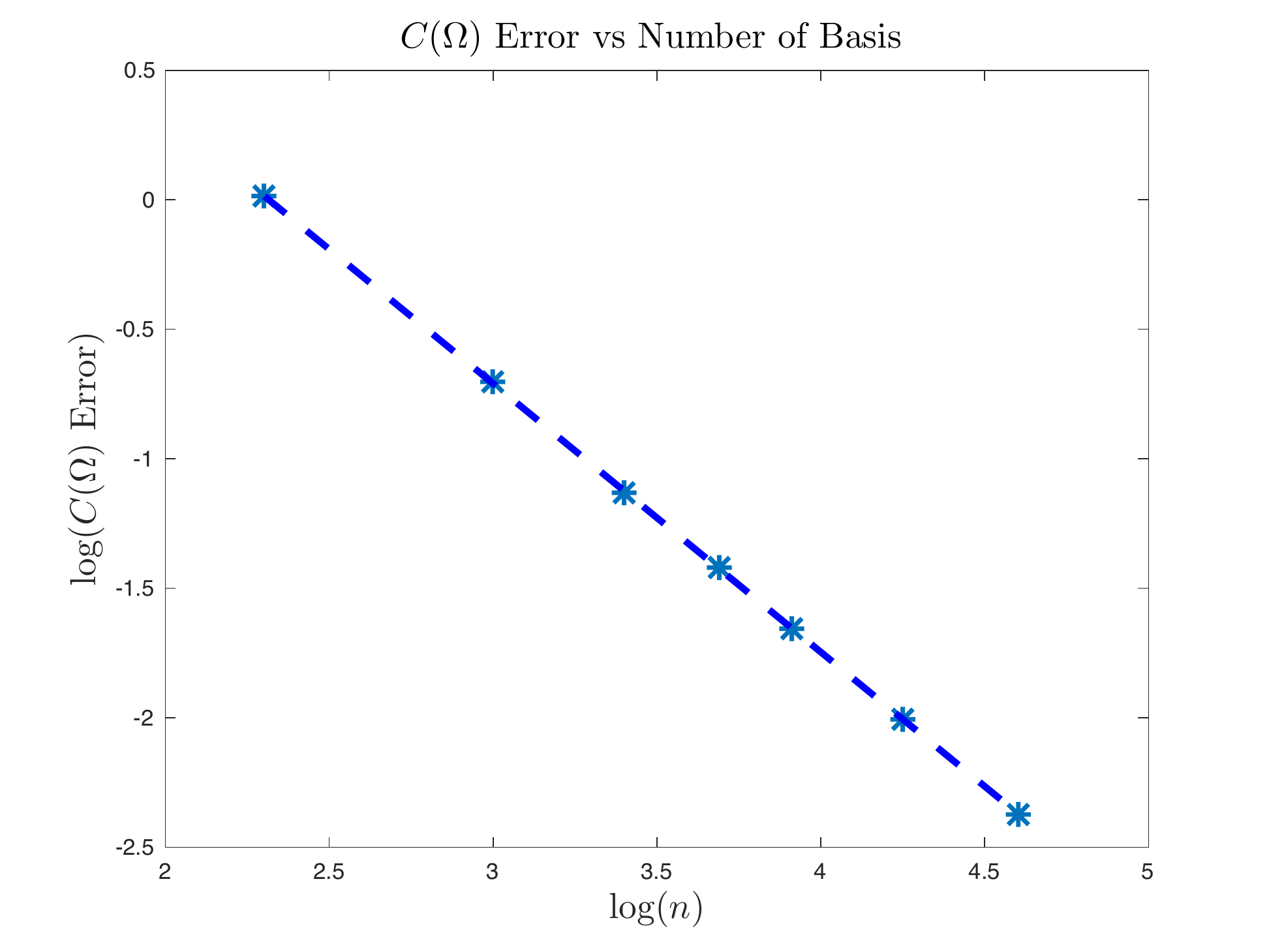}\\
\end{tabular}
\caption{Convergence rates using Gaussian kernel for synthetic data}
\label{fig:logsup}
\end{figure}
The plots in Fig.\ref{fig:logsup} show the rate of convergence of $L^2$ error and the $C(\Omega)$ error with respect to the number of basis functions. The {\em{log}} along the axes in the figures refer to the natural logarithm unless explicitly specified.

\subsection{Experimental Road Profile Data}
The road profile to be estimated in this subsection is based on the experimental data obtained from the Vehicle Terrain Measurement System shown in Fig.~\ref{fig:circle}.  The constants in the estimation problem are initialized to the same numerical values as in previous subsection.
\begin{figure}[h]
\centering
\begin{tabular}{cc}
\includegraphics[width=0.4\textwidth]{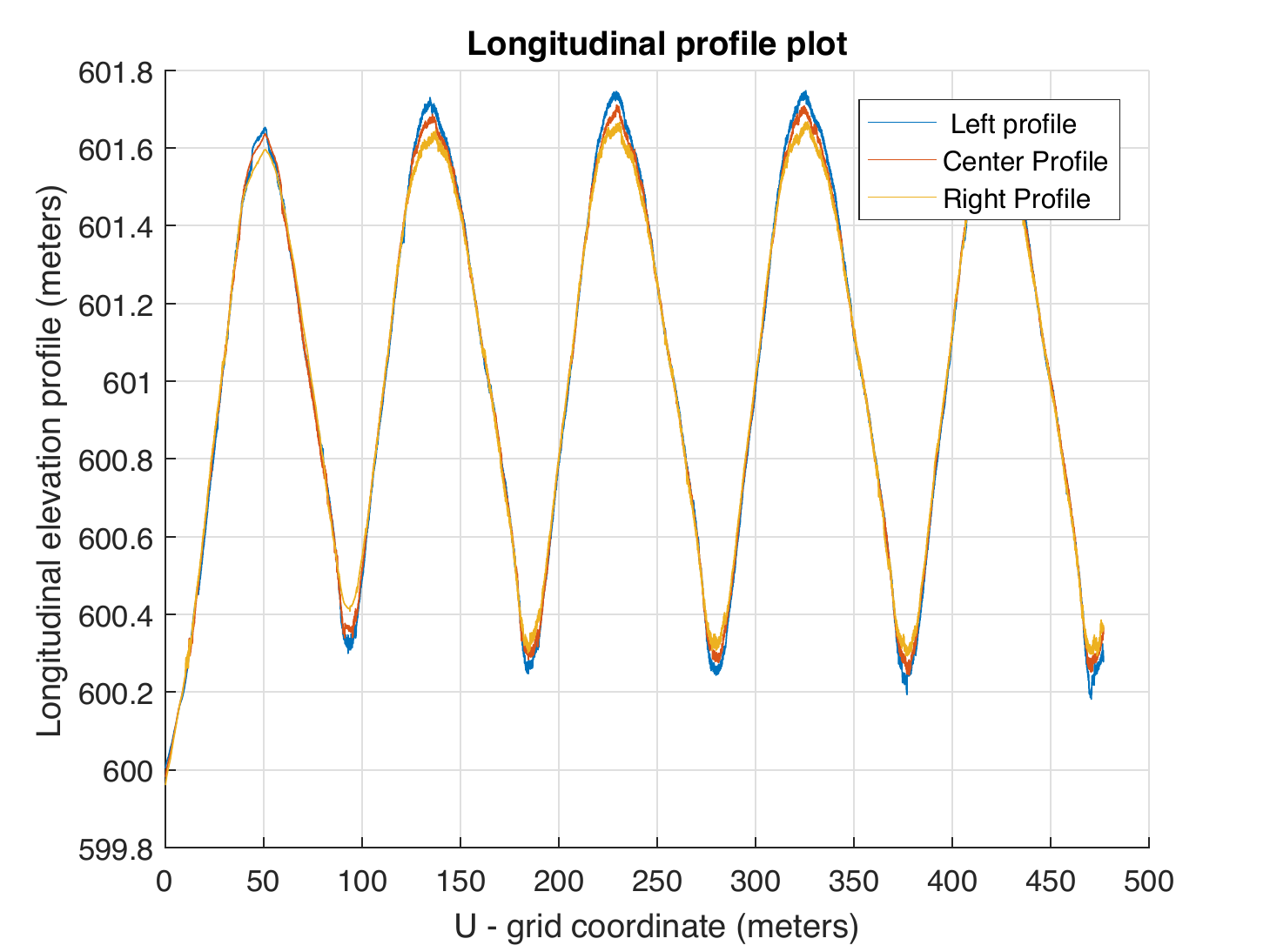}
&
\includegraphics[width=0.4\textwidth]{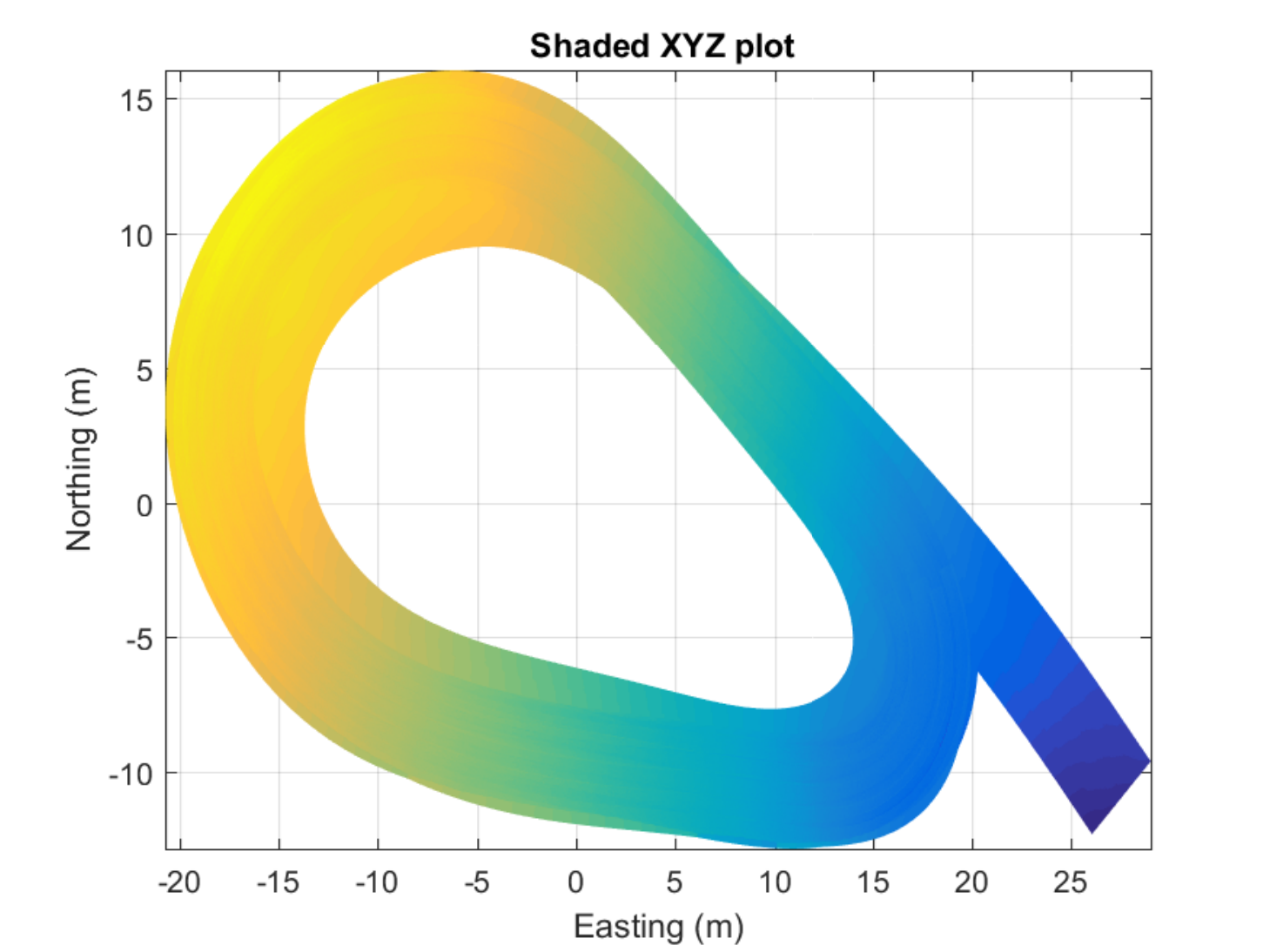}\\
{Longitudinal Elevation Profile.}
&
{Circular Path followed by VTMS.}
\end{tabular}
\caption{Experimental Data From VTMS.}
\label{fig:circle}
\end{figure}
In the first study in this section the adaptive estimation problem is formulated in the RKHS $H = \overline{k_x(\cdot)|x\in \Omega\}}$ with  $k_x(\cdot)=e^\frac{-\|x-{\cdot}\|^2}{2\sigma^2 }$. The radial basis functions, each with standard deviation of $\sigma=50$, span over the range of with a collection of centers located at $\{s_i\}_{i=1}^{n}\subseteq \Omega$ evenly separated along the arclength. This is  repeated for kernels defined using B-splines of first order and second order respectively. 

Fig.\ref{fig:Kernels} shows the finite dimensional estimates of the road and the true road surface $f$ for a data representing single lap around the circular track, the finite dimensional estimates $\hat{f}_n$ are plotted for different number of basis kernels ranging from $n=\{35,50,\cdots,140\}$ using the Gaussian kernel as well as the second order B-splines. 
The finite dimensional estimates $\hat{f}_n$ of the road profile and the true road profile $f$ for data collected representing multiple laps around the circular track is plotted for the first order B-splines as shown in Fig.~\ref{fig:Lsplines Road}. The plots in Fig.~\ref{fig:sup_error_compare} show the rate of convergence of the $L^2$ error and the $C(\Omega)$ error with respect to number of basis functions. 
% From table ~\ref{table:slopes}
It is seen that the rate of convergence for $2^{nd}$ order B-Spline is better as compared to other kernels used to estimate in these examples. This corroborates the fact that smoother kernels are expected to have better convergence rates. 

Also, the condition number of the Grammian matrix varies with $n$, as illustrated in Table.\ref{table:1} and Fig.\ref{fig:conditionnumber}. This is an important factor to consider when choosing a specific kernel for the RKHS embedding technique since it is well known that the error in numerical estimates of solutions to linear systems is bounded above by the condition number. The implementation of the RKHS embedding method requires such a solution that depends on the grammian matrix of the kernel bases at each time step. We see that the condition number of Grammian matrices for exponentials is $\mathcal{O}(10^{16})$ greater than the corresponding matrices for splines. Since the sensitivity of the solutions of linear equations is bounded by the condition numbers, it is expected that the use of exponentials could suffer from a severe loss of accuracy as the dimensionality increases. The development for preconditioning techniques for Grammian matrices constructed from radial basis functions to address this problem is an area of active research.
% \begin{center}
% \centering
% \begin{table}[h!]
% \centering
% \begin{tabular}{|p{2cm}|p{2cm}|p{2cm}|}
% \hline
% Type of Kernel& $L^2$ Convergence rate (-ve)& $C(\Omega)$ Convergence rate (-ve)\\ 
%  \hline \hline
%  Gaussian & 0.3804 & 0.3804 \\ 
% $1^{st} \mathcal{o}$ B-S& 0.5763& 0.5763 \\
% $2^{nd} \mathcal{o}$  B-S& 0.8195 & 0.8195\\
%  \hline
% \end{tabular}
% \caption{Condition number of Grammian Matrix vs Number of Basis Functions}
% \label{table:slopes}
% \end{table}
% \end{center}
\begin{figure}[H]
\centering
\begin{tabular}{cc}
\includegraphics[width = 0.4 \textwidth]{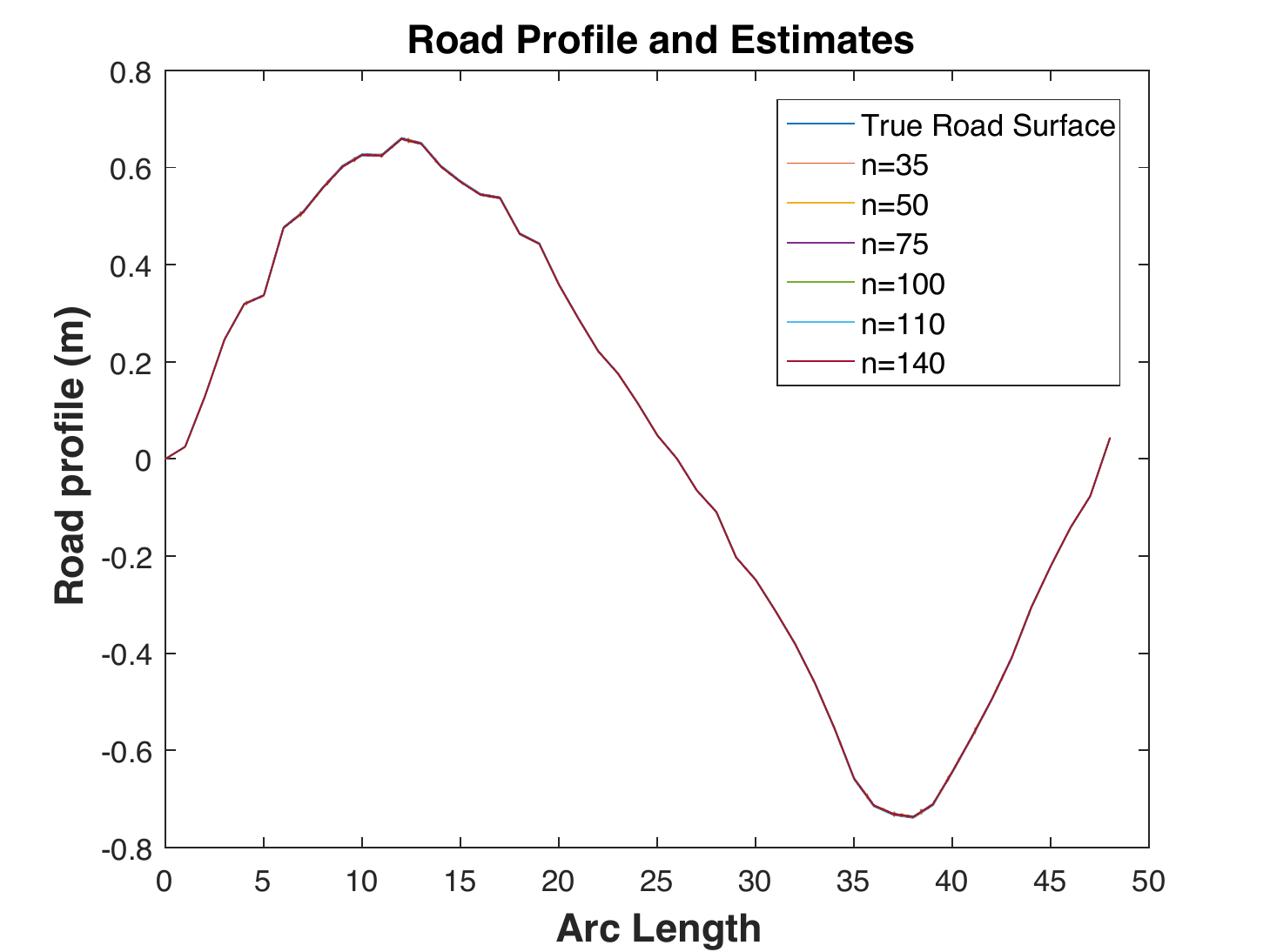}
&
\includegraphics[width = 0.4 \textwidth]{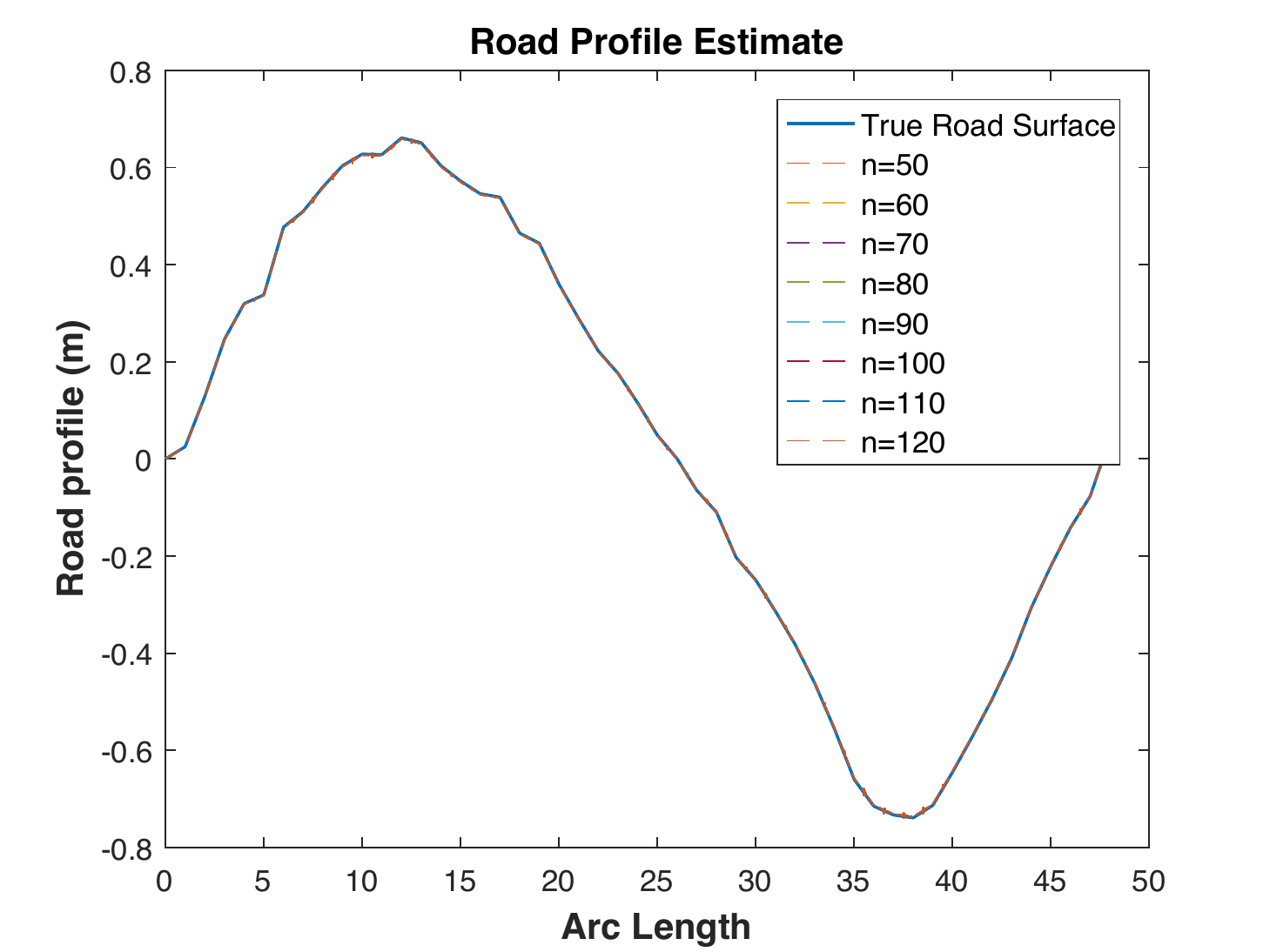}\\
{Road surface estimates for Gaussian kernels}
&
{Road surface estimate for second-order B-splines}
\end{tabular}
\caption{Road surface estimates for single lap}
\label{fig:Kernels}
\end{figure}

\begin{figure}[H]
\centering
\includegraphics[width = 0.4 \textwidth]{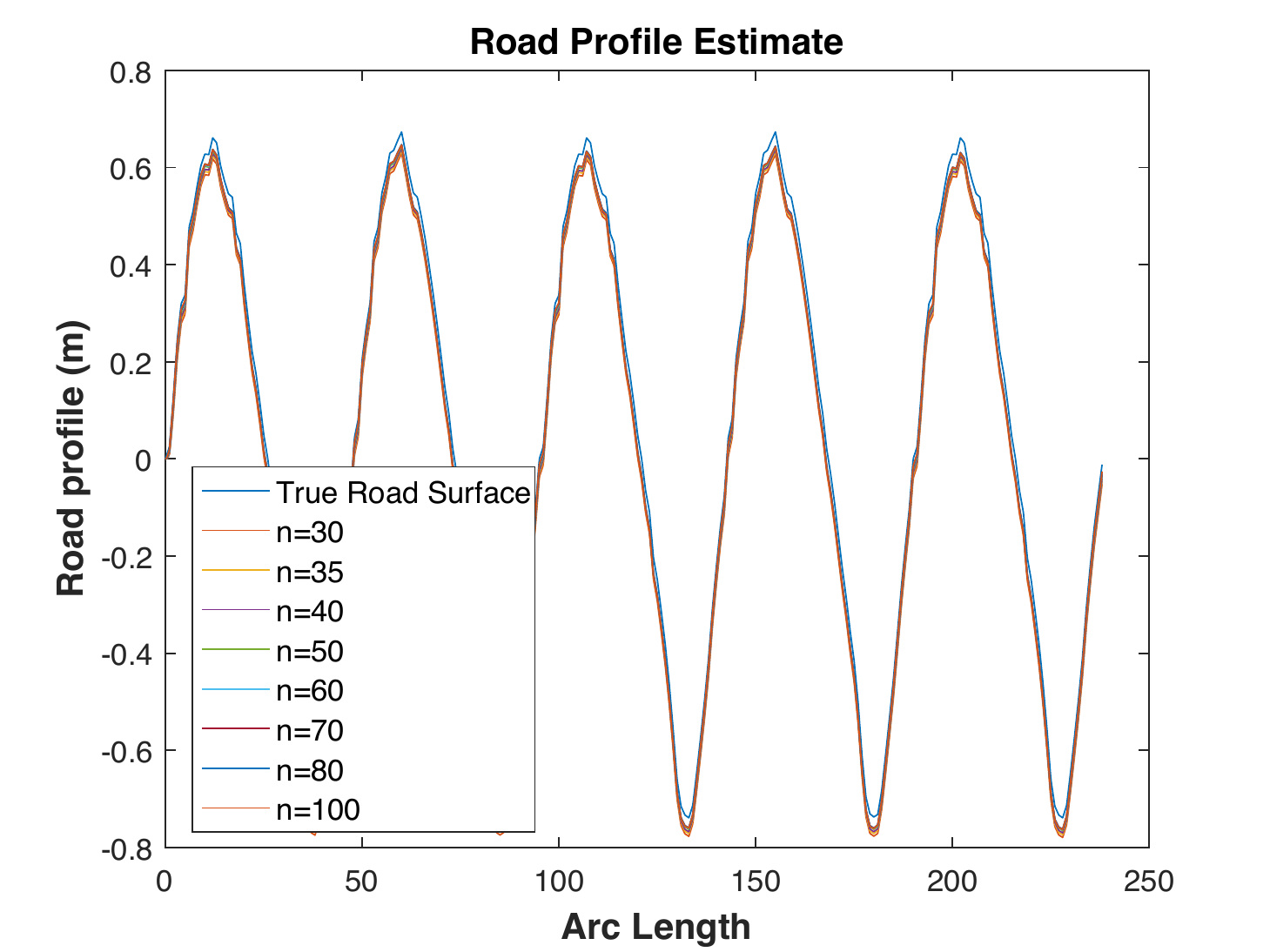}
\caption{Road surface estimate using first-order B-splines}
\label{fig:Lsplines Road}
\end{figure}
\begin{center}
\begin{figure}[H]
\centering
\begin{tabular}{cc}
\includegraphics[width = 0.5\textwidth]{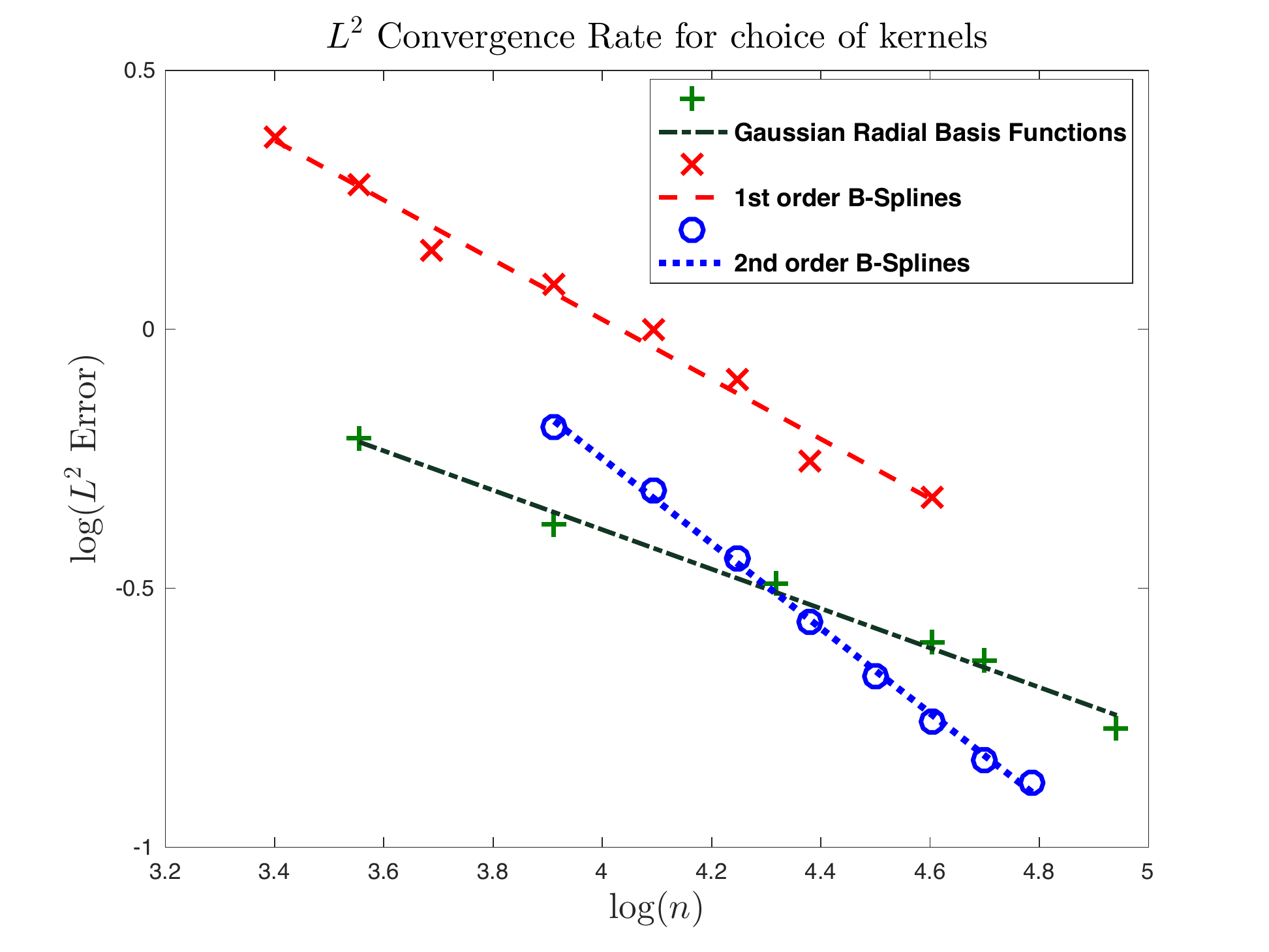}
&
\includegraphics[width = 0.5\textwidth]{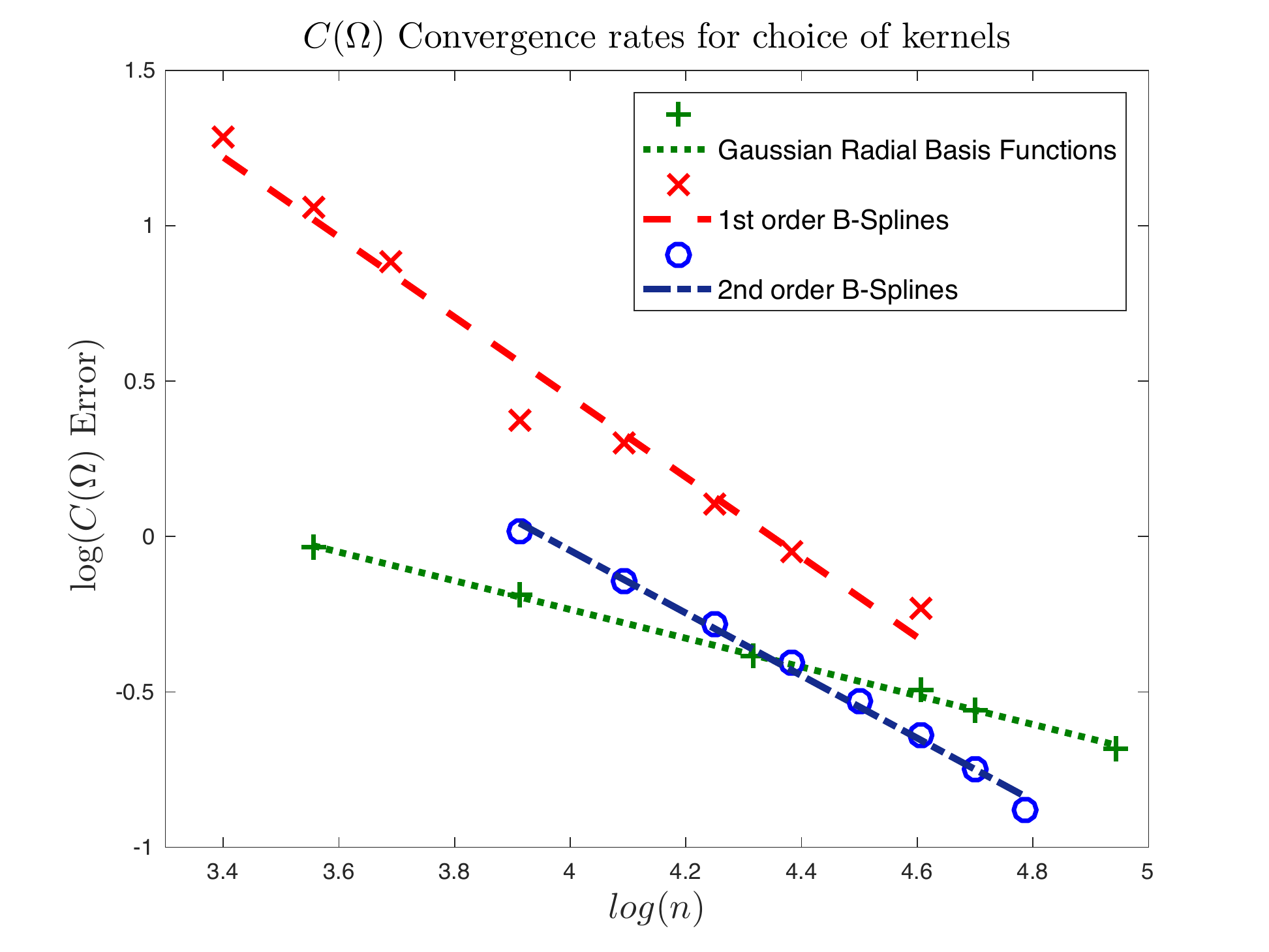}\\
\end{tabular}
\caption{Convergence rates for different kernels}
\label{fig:sup_error_compare}
\end{figure}
\end{center}
\begin{center}
\centering
\begin{table}[H]
\centering
\begin{tabular}{|p{1cm}|p{2.2cm}|p{2.2cm}|p{2.2cm}|}
\hline
No. of Basis Functions &  Condition No. (First order B-Splines) $\times 10^3$ & Condition No.(Second order B-Splines) $\times 10^4$ & Condition No.(Gaussian Kernels) $\times 10^{20}$\\ 
 \hline \hline
 10 & 0.6646 & 0.3882 & 0.0001 \\ 
 20 & 1.0396 & 0.9336 & 0.0017  \\
 30 & 1.4077 & 1.5045 & 0.0029  \\
 40 & 1.7737 & 2.0784 & 0.0074 \\
 50 & 2.1388 & 2.6535 & 0.0167\\ 
 60 & 2.5035 & 3.2293 & 0.0102\\ 
 70 & 2.8678 & 3.8054& 0.0542\\ 
 80 & 3.2321 & 4.3818& 0.0571\\ 
 90 & 3.5962 & 4.9583& 0.7624\\ 
 100 & 3.9602 & 5.5350& 1.3630\\ 
 \hline
\end{tabular}
\caption{Condition number of Grammian Matrix vs Number of Basis Functions}
\label{table:1}
\end{table}
\end{center}
\begin{figure}[H]
\centering
\includegraphics[height=0.3\textheight,width=0.65\textwidth]{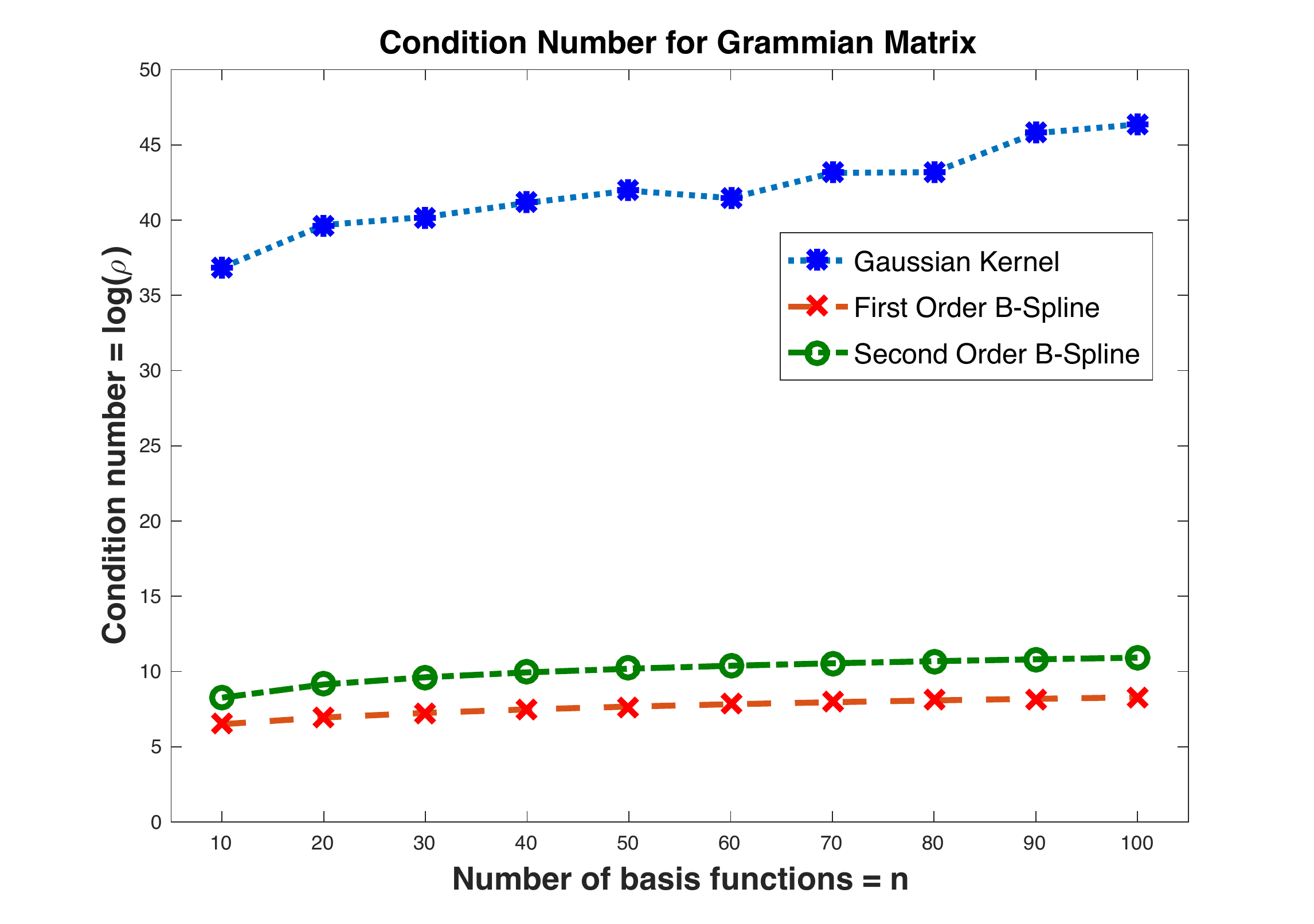}
\caption{Condition Number of Grammian Matrix vs Number of Basis Functions}
\label{fig:conditionnumber}
\end{figure}
\vspace{-1cm}
	\section{Conclusions}
	\label{sec:conclusions}
    In this paper, we introduced a novel framework based on the use of RKHS embedding to study online adaptive estimation problems. The applicability of this framework to solve estimation problems that involve high dimensional scattered data approximation provides the motivation for the theory and algorithms described in this paper. A quick overview of the background theory on RKHS enables rigorous derivation of the results in  Sections \ref{sec:existence} and \ref{sec:finite}. In this paper we derive (1) the  sufficient conditions for the existence and uniqueness of solutions to the RKHS embedding problem, (2) the  stability and convergence of the state estimation error, and (3) the convergence of the finite dimensional approximate solutions  to  the solution of the  infinite dimensional state space. To illustrate the utility of this approach, a simplified numerical example of adaptive estimation of a road profile is studied and the results are critically analyzed. It would be of further interest to see the ramifications of using multiscale kernels to achieve semi-optimal convergence rates for functions in a scale of Sobolev spaces. It would likewise be important to extend this framework  to adaptive control problems and examine the consequences of {\em persistency of excitation} conditions in the RKHS setting, and further extend the approach to adaptively generate bases over the state space. 
    
    \bibliography{rkhs.bib} 

\begin{thebibliography}{10}

\bibitem{wendland}
Holger Wendland.
\newblock {\em Scattered data approximation}.
\newblock Cambridge University Press, 2005.

\bibitem{bsdr1997}
M.A.~Demetriou J.~Baumeister, W.~Scondo and I.G. Rosen.
\newblock On-line parameter estimation for infinite dimensional dynamical
  systems.
\newblock {\em SIAM Journal of Control and Optimisation}, 1997.

\bibitem{bdrr1998}
S.~Reich M.~Bohm, M.A.~Demetriou and I.G. Rosen.
\newblock Model reference adaptive control of distributed parameter systems.
\newblock {\em SIAM Journal of Control and Optimisation}, 1998.

\bibitem{dzcf2012}
Y.~Chen W.~Dong, Y.~Zhao and J.A. Farrell.
\newblock Tracking control for nonaffine systems: A self-organizing
  approximation approach.
\newblock {\em IEEE Transactions on Neural Networks and Learning Systems},
  2012.

\bibitem{Thrun2005Probabilistic}
Sebastian Thrun, Wolfram Burgard, and Dieter Fox.
\newblock {\em Probabilistic Robotics}.
\newblock 2005.

\bibitem{Whyte2006SLAM1}
H.~Durrant-Whyte and T.~Bailey.
\newblock Simultaneous localization and mapping: Part {I}.
\newblock {\em IEEE Robotics Automation Magazine}, 13(2):99--110, June 2006.

\bibitem{Whyte2006SLAM2}
T.~Bailey and H.~Durrant-Whyte.
\newblock Simultaneous localization and mapping ({SLAM}): part {II}.
\newblock {\em IEEE Robotics Automation Magazine}, 13(3):108--117, Sept 2006.

\bibitem{Dissanayake2011Review}
G.~Dissanayake, S.~Huang, Z.~Wang, and R.~Ranasinghe.
\newblock A review of recent developments in simultaneous localization and
  mapping.
\newblock {\em 6th International Conference on Industrial and Information
  Systems}, 2011.

\bibitem{Dissanayake2000Computational}
G.~Dissanayake, H.~Durrant-Whyte, and T.~Bailey.
\newblock A computationally efficient solution to the simultaneous localisation
  and map building ({SLAM}) problem.
\newblock In {\em Proceedings 2000 ICRA. Millennium Conference}, 2000.

\bibitem{Huang2007Convergence}
Shoudong Huang and Gamini Dissanayake.
\newblock Convergence and consistency analysis for extended kalman filter based
  {SLAM}.
\newblock {\em Transaction on Robotics}, 23(5):1036--1049, October 2007.

\bibitem{Julier2001Counter}
S.~J. Julier and J.~K. Uhlmann.
\newblock A counter example to the theory of simultaneous localization and map
  building.
\newblock In {\em Proceedings 2001 IEEE International Conference on Robotics
  and Automation}, 2001.

\bibitem{Meyer}
Yves Meyer.
\newblock {\em Wavelets and operators}.
\newblock Cambridge University Press, 1992.

\bibitem{mallat}
Stephane Mallat.
\newblock {\em A wavelet tour of signal processing}.
\newblock Academic Press, 1999.

\bibitem{daubechies}
Ingrid Daubechies.
\newblock {\em Ten Lectures on Wavelets}.
\newblock SIAM, 1992.

\bibitem{dl1993}
Ronald~A. DeVore and George Lorentz.
\newblock {\em Constructive Approximation}.
\newblock Springer-Verlag, 1993.

\bibitem{opfer1}
Roland Opfer.
\newblock Tight frame expansions of multiscale reproducing kernels in sobolev
  spaces.
\newblock {\em Applied Computational Harmonic Analysis}, 2006.

\bibitem{opfer2}
Roland Opfer.
\newblock Multiscale kernels.
\newblock {\em Advances in Computational Mathematics}, 2006.

\bibitem{dkpt2006}
Dominque~Picard Ronald~DeVore, Gerard~Kerkyacharian and Vladimir Temlyakov.
\newblock Approximation methods for supervised learning.
\newblock {\em Foundations of Computational Mathematics}, 2006.

\bibitem{kt2007}
S.V. Konyagin and V.N. Temlyakov.
\newblock The entropy in learning theory. {Error Estimates}.
\newblock {\em Constructive Approximation}, 25(1):1--27, 2007.

\bibitem{cdkp2001}
Gerard~Kerkyacharian Albert~Cohen, Ronald~DeVore and Dominique Picard.
\newblock Maximal spaces with given rate of convergence for thresholding
  algorithms.
\newblock {\em Applied and Computational Harmonic Analysis}, 2001.

\bibitem{t2008}
V.N. Temlyakov.
\newblock Approximation in learning theory.
\newblock {\em Constructive Approximation}, 2008.

\bibitem{sb2012}
Shankar Sastry and Marc Bodson.
\newblock {\em Adaptive Control: Stability, Convergence and Robustness}.
\newblock Dover, 2011.

\bibitem{IaSu}
Petros~A. Ioannou and Jing Sun.
\newblock {\em Robust Adaptive Control}.
\newblock Dover, 2012.

\bibitem{PoFar}
Jay~A. Farrell and Marios~M. Polycarpou.
\newblock {\em Adaptive approximation based control: unifying neural, fuzzy and
  traditional adaptive approximation approaches}.
\newblock Wiley, 2006.

\bibitem{dmp1994}
B.~Maslowski T.E.~Duncan and B.~Pasik-Duncan.
\newblock Adaptive boundary and point control of linear stochastic distributed
  parameter systems.
\newblock {\em SIAM J. Control Optim.}, 1997.

\bibitem{dp1988}
T.E. Duncan and B.~Pasik-Duncan.
\newblock Adaptive control of linear delay time systems.
\newblock {\em Stochastics}, 1988.

\bibitem{dpg1991}
B.~Pasik-Duncan T.E.~Duncan and B.~Goldys.
\newblock Adaptive control of linear stochastic evolution systems.
\newblock {\em Stochastics Rep.}, 1991.

\bibitem{p1992}
B.~Pasik-Duncan.
\newblock On the consistency of a least squares identification procedure in
  linear evolution systems.
\newblock {\em Stochastics Report.}, 1992.

\bibitem{HC}
N.~Hovakimyan and C.~Cao.
\newblock {\em $\mathcal{L}^1$ Adaptive Control Theory}.
\newblock SIAM, 2010.

\bibitem{naranna}
K.S Narendra and A.M.Annaswamy.
\newblock {\em Stable Adaptive Systems}.
\newblock Prentice Hall, 1989.

\bibitem{NarPar199D}
K.S Narendra and K.Parthasarthy.
\newblock Identification and control of dynamical systems using neural
  networks.
\newblock {\em IEEE Trans. Neural Networks}, 1990.

\bibitem{narkud}
K.~S. Narendra and P.~Kudva.
\newblock Stable adaptive schemes for system identification and control - {Part
  II}.
\newblock {\em IEEE Transactions on Systems, Man, and Cybernetics},
  SMC-4(6):552--560, Nov 1974.

\bibitem{mornar}
A.P.Morgan and K.S Narendra.
\newblock On stability of nonautonomous differential equations
  $\dot{x}=[a+b(t)]x$, with skew symmetric $b(t)$.
\newblock {\em SIAM Journal of Control and Optimisation}, 1977.

\bibitem{bk1989}
H.T. Banks and K.~Kunisch.
\newblock {\em Estimation Techniques for Distributed Parameter Systems}.
\newblock Birkhauser, 1989.

\bibitem{sz2007}
S.~Smale and X.~Zhou.
\newblock Learning theory estimates via integral operators and their
  approximations.
\newblock {\em Constructive Approximation}, 2007.

\bibitem{devore1998}
Ronald~A. DeVore.
\newblock Adapting to unknown smoothness via wavelet shrinkage.
\newblock {\em Acta Numerica}, 1998.

\bibitem{af2003}
Adams~R. A. and Fournier John.
\newblock {\em Sobolev spaces}.
\newblock Elsevier, 2003.

\bibitem{pazy}
Amnon Pazy.
\newblock {\em Semigroups of Linear Operators and Applications to Partial
  Differential Equations}.
\newblock Springer, 2011.

\bibitem{Farkas2016Variations}
Balint Farkas and Sven ake Wegner.
\newblock Variations on barbalat's lemma.
\newblock {\em Arxiv:1411.1611v3}, 2016.

\bibitem{d1993}
M.A. Demetriou.
\newblock {\em Adaptive Parameter Estimation of Abstract Parabolic and
  Hyperbolic Distributed Parameter Systems}.
\newblock PhD thesis, University of Southern California, 1993.

\bibitem{dr1994}
M.A. Demetriou and I.G. Rosen.
\newblock Adaptive identification of second order distributed parameter
  systems.
\newblock {\em Inverse Problems}, 1994.

\bibitem{dr1994pe}
M.A. Demetriou and I.G. Rosen.
\newblock On the persistence of excitation in the adaptive identification of
  distributed parameter systems.
\newblock {\em IEEE Transactions of Automatic Control}, 1994.

\bibitem{kr1994}
Joseph Kazimir and I.G. Rosen.
\newblock Adaptive estimation of nonlinear distributed parameter systems.
\newblock {\em International Series of Numerical Mathematics, Birkhauser
  Verlag}, 1994.

\end{thebibliography}
    \bibliographystyle{unsrt}
\end{document}